\documentclass[11pt]{article}
\usepackage[margin=1in]{geometry} 
\usepackage{algorithmicx}

\usepackage{pifont}

\def\david#1{\marginpar{$\leftarrow$\fbox{D}}\footnote{$\Rightarrow$~{\sf\textcolor{blue}{#1 --David}}}}

\def\ren#1{\marginpar{$\leftarrow$\fbox{R}}\footnote{$\Rightarrow$~{\sf\textcolor{violet}{#1 --Renata}}}}

\usepackage[T1]{fontenc}
\usepackage[utf8]{inputenc}
\usepackage{authblk}
\usepackage{amsmath, amssymb, amsthm, thmtools, amsfonts, bm, bbm, thm-restate}
\usepackage{algorithm}
\usepackage{algorithmicx}
\usepackage[noend]{algpseudocode}
\usepackage[noadjust]{cite}
\usepackage{pdfpages}

\usepackage{asymptote}
\usepackage{graphicx}
\usepackage[disable]{todonotes}
\usepackage{multirow}
\usepackage{comment}
\usepackage{dsfont}
\usepackage{bigstrut}
\usepackage{caption}
\usepackage{subcaption}
\usepackage{multirow}
\usepackage{makecell}
\usepackage{booktabs}
\usepackage{xcolor}
\usepackage{enumerate}

\usepackage{commath}

\definecolor{BrickRed}{rgb}{0.8,0.25,0.33}
\usepackage[pagebackref,colorlinks,citecolor=blue,linkcolor=BrickRed, hypertexnames=false]{hyperref}

\usepackage{enumitem}
\usepackage[normalem]{ulem}

\usepackage[framemethod=tikz]{mdframed}
\usepackage{nicefrac}

\usepackage{tikz}
\usetikzlibrary{positioning, fit, calc}

\usepackage{pifont}

\usepackage{cleveref}

\newtheorem{thm}{Theorem}[section]
\newtheorem{cor}[thm]{Corollary}
\newtheorem{prop}[thm]{Proposition}

\newtheorem{fact}[thm]{Fact}
\newtheorem{lem}[thm]{Lemma}
\newtheorem{lemma}[thm]{Lemma}

\newtheorem{Def}[thm]{Definition}
\newtheorem{obs}[thm]{Observation}

\newtheorem{rem}[thm]{Remark}

\crefname{thm}{Theorem}{Theorems}
\crefname{cla}{Claim}{Claims}
\crefname{lem}{Lemma}{Lemmas}
\crefname{fact}{Fact}{Facts}
\crefname{prop}{Proposition}{Propositions}

\allowdisplaybreaks

\newcommand{\veps}{\varepsilon}

\newcommand{\E}{\mathbb{E}}
\newcommand{\R}{\mathbb{R}}
\renewcommand{\Pr}{\mathds{P}r}

\newcommand{\eps}{\varepsilon}
\newcommand{\calA}{\mathcal{A}}
\newcommand{\calB}{\mathcal{B}}
\newcommand{\calD}{\mathcal{D}}

\newcommand{\calP}{\mathcal{P}}
\newcommand{\calR}{\mathcal{R}}

\newcommand{\poly}{\mathrm{poly}}
\newcommand{\Exp}{\mathrm{Exp}}

\newcommand{\Uni}{\mathrm{Uni}}

\newcommand{\Cov}{\mathrm{Cov}}

\newcommand{\constantC}{5}
\newcommand{\expM}{5}
\newcommand{\goodConst}{920}

\newcommand{\m}{2^{\constantC}k^{\expM}\alpha^{\expM}}

\newcommand{\by}{\mathbf{y}}

\let\vec\mathbf
\renewcommand{\vec}{\mathbf}

\title{Dimension-Free  Correlated Sampling for the Hypersimplex}
\date{}

\author[1]{Joseph (Seffi) Naor\thanks{Supported in part by ISF grant 3001/24 and United States - Israel
BSF grant 2022418.}}
\author[2]{Nitya Raju}
\author[3]{Abhishek Shetty \thanks{Supported in part by ARO award W911NF-21-1-0328, the Simons Foundation, NSF award DMS-2031883, a DARPA AIQ award, an NSF FODSI Postdoctoral Fellowship and an Apple AI/ML Fellowship.}}
\author[2]{Aravind Srinivasan\thanks{Supported in part by NSF award number CCF-1918749.}}
\author[2]{Renata Valieva\thanks{Supported in part by NSF award number CCF-1918749 and CNS-2317194.}}
\author[1]{David Wajc\thanks{Supported in part by the Taub Family Foundation ``Leader in Science and Technology'' fellowship, Grand Technion Energy Program (GTEP) and ISF grant 3200/24.}}
\affil[1]{Technion —-- Israel Institute of Technology}
\affil[2]{University of Maryland}
\affil[3]{Massachusetts Institute of Technology}

\usepackage{relsize}

\date{\vspace{-1cm}}

\begin{document}

\maketitle

\pagenumbering{gobble}
\begin{abstract}
    Sampling from multiple distributions so as to maximize overlap has been studied by statisticians since the 1950s. 
    Since the 2000s, such correlated sampling from the probability simplex has been a powerful building block in disparate areas of theoretical computer science. 
    We study a generalization of this problem to sampling sets from given vectors in the hypersimplex, i.e., outputting sets of size (at most) some $k$ in $[n]$, while maximizing the sampled sets' overlap.
    Specifically, the expected difference between two output sets should be at most $\alpha$ times their input vectors' $\ell_1$ distance.
    A value of $\alpha=O(\log n)$ is known to be achievable, due to Chen et al.~(ICALP'17). 
    We improve this factor to $O(\log k)$, independent of the ambient dimension~$n$.
    Our algorithm satisfies other desirable properties, including (up to a $\log^* n$ factor) input-sparsity sampling time, logarithmic parallel depth and dynamic update time, as well as preservation of submodular objectives. 
    Anticipating broader use of correlated sampling algorithms for the hypersimplex, we present applications of our algorithm to online paging, offline approximation of metric multi-labeling and swift multi-scenario submodular welfare approximating reallocation.
\end{abstract}

{\smaller[0.5]
\tableofcontents
}

\newpage

\pagenumbering{arabic}
\section{Introduction}

Consider the following natural \emph{correlated sampling} problem: 
we wish to sample from numerous discrete distributions over $[n]$, or equivalently round numerous points in the $n$-dimensional probability simplex, $\Delta_{n}\triangleq \{\vec{x}\in [0,1]^n:~  \norm{\vec{x}}_1 = 1\}$, in a \emph{consistent manner}. In particular, for any two distributions with low total variation distance, i.e., points in $\Delta_{n}$ with low $\ell_1$ distance, $\norm{\vec{x}-\vec{y}}_1 \triangleq \sum_i |x_i-y_i|$, we wish the sets drawn according to these points to be similar in expectation.

This problem is well-understood \cite{kleinberg2002approximation,holenstein2007parallel,ge2011geometric,buchbinder2018simplex,angel2019pairwise,bavarian2020optimality}. Moreover, efficient correlated sampling algorithms have found numerous applications to disparate areas of computation, from approximation algorithms \cite{kleinberg2002approximation,ge2011geometric,buchbinder2018simplex,sharma2014multiway}, parallel repetition \cite{holenstein2007parallel,barak2008rounding}, replicability and privacy \cite{badih2021user,bun2023stability}, parallel sampling \cite{liu2022simple,anari2024parallel}, cryptography \cite{rivest2016symmetric}, locality sensitive hashing  \cite{broder1997resemblance,charikar2002similarity},  and more.

Such correlated sampling is a special case of what statisticians refer to as \emph{coordinated sampling}, and have been studying since the 1950s \cite{keyfitz1951sampling}. Here, the motivation of maximizing overlap is in minimizing overhead of initiation interviews for newly sampled interviewees or overhead of hiring new interviewers for new regions with sampled interviewees.
Such rationale was used and adopted, for example in the U.S.\ Bureau of the Census household surveys.
See discussion of such applications and nearly half a century of work on such questions in the survey by Ernst~\cite{ernst1999maximization}.

In this work we study the following correlated sampling problem, introduced by \cite{chen2017correlated}, generalizing from the probability simplex to the convex hull of the hypersimplex and the origin, i.e., the polytope whose extreme points are vectors in $\{0,1\}^n$ with at most $k$ ones,
$$\Delta_{n,k}\triangleq \{ \vec{x}\in [0,1]^n:~\norm{\vec{x}}_1 \leq k\}.$$
We wish to round points in this set while preserving marginals, respecting the cardinality constraint of $k$, and guaranteeing that the expected distance of the output sets for any two points $\vec{x},\vec{y}\in \Delta_{n,k}$ is bounded in terms of the points' $\ell_1$ distance. Formally, we study the following problem.

\begin{restatable}{Def}{correlatedsamplingdef}
    \label{def:correlated-sampling}
A distribution $\calD$ over deterministic algorithms $\calA:\Delta_{n,k}\to 2^{[n]}$ mapping points in $\Delta_{n,k}$ to subsets of $[n]$ of cardinality at most $k$ is an \textbf{$\boldsymbol{\alpha}$-stretch correlated sampling algorithm for $\boldsymbol{\Delta_{n,k}}$} if it satisfies the following properties for all~$\vec{x},\vec{y}\in \Delta_{n,k}$:
\begin{enumerate}[label=(P{{\arabic*}})]
    \item \label{prop:k-uniform} \textbf{(Cardinality)} 
    $|\calA(\vec{x})| \in \left\{\lfloor \|\vec{x}\|_1 \rfloor,\lceil \|\vec{x}\|_1 \rceil \right\}$ (hence $|\calA(\vec{x})| \leq k$) for all $\calA$ in the support~of~$\calD$.
    \item \label{prop:marginals} \textbf{(Marginals)} $\Pr_{\calA\sim \calD}[i\in \calA(\vec{x})]=x_i$ for all $i\in [n]$.
    \item \label{prop:distance}\textbf{(Stretch)} $\E_{\calA\sim\calD}[\;| \calA(\vec{x})\oplus  \calA(\vec{y})|\;]\leq \alpha \cdot \norm{\vec{x}-\vec{y}}_1$.
\end{enumerate}
\end{restatable}

To clarify Property \ref{prop:distance} and its use: after a common initialization step where we draw an algorithm $\calA\sim\calD$, 
we guarantee that when rounding any two vectors $\vec{x},\vec{y}\in \Delta_{n,k}$, the output's expected distance (over the random choice of $\calA$) is at most the vectors' $\ell_1$ distance times some $\alpha$.
Naturally, we wish to keep $\alpha$  as small as possible. 

For $k=1$ it is known that $\alpha$ must be at least two \cite{angel2019pairwise,bavarian2020optimality},\footnote{More precisely, some $\vec{x},\vec{y}$ must have stretch at least $2/(1+\delta)$, for $\delta\triangleq \frac{1}{2}\norm{\vec{x}-\vec{y}}_1$, though $\delta$ may tend to zero.} and this bound is attainable, provided $\norm{x}_1=k = 1$, i.e., for $\Delta_n\subsetneq \Delta_{n,1}$ \cite{kleinberg2002approximation,holenstein2007parallel,ge2011geometric,buchbinder2018simplex,angel2019pairwise,liu2022simple}. (As we show in \Cref{sec:correlated-sampling-appendix}, a constant bound of $\alpha=3$ for $\Delta_{n,1}$ more generally follows from the bound~for~$\Delta_n$.) 
For $k > 1$, in contrast, \cite{chen2017correlated} show that an $O(\log n)$ stretch is achievable, and use it for applications to metric multi-labelling, which we define later.

The stark difference between the case $k=1$ and larger $k>1$ should be apparent. 
The natural question, which we study, is whether this logarithmic in $n$ stretch is inherent to the problem.

\subsection{Our Contribution}

Our main contribution is a correlated sampling algorithm for the hypersimplex with $k>1$ with stretch logarithmic in $k$, but independent of the input dimension, $n$.

\begin{restatable}{thm}{correlatedsampling}\label{thm:correlated-sampling}
    There exists an $O(\log k)$-stretch correlated sampling algorithm for $\Delta_{n,k}$ for all $n$.
\end{restatable}
Our algorithm is extremely fast, and can be implemented  in $O(nnz \cdot \log^* n)$  time, for $nnz\leq n$ the number of nonzeros and $\log^*n$ the (extremely slowly-growing) iterated logarithm function (see \Cref{def:log*}). Thus, the algorithm's running time is also \emph{very nearly} linear in the input sparsity, and dimension independent. Our algorithm can further be implemented in $O(\log n\cdot \log^*n)$ parallel depth and $O(\log (nnz)\cdot \log^*n)$ update time in dynamic settings.

Property \ref{prop:marginals} immediately implies preservation of linear objectives. We also prove that our algorithm preserves \emph{submodular} objectives. (See \Cref{sec:submodularity} for some background, but for now suffice it to say that these functions capture the ubiquitous phenomenon of diminishing returns.) In particular, we prove that our correlated sampling algorithm's output has at least as high a submodular value as independent rounding with marginals $\vec{x}$ (the so-called \emph{multilinear extension}), while being oblivious to the objective submodular function, and also satisfying the cardinality constraint and desired low stretch guarantees. In \Cref{sec:partition} we extend this property to correlated sampling for partition matroids , where the ground set is partitioned into parts, with a cardinality constraint on each part.

Due to the widespread uses of correlated sampling for the probability simplex, we anticipate numerous subsequent applications of our correlated sampling algorithm for the hypersimplex.
To illustrate this, we provide several simple applications of our algorithm, to online paging, offline metric multi-labeling, and dynamic submodular welfare maximizing reallocation.
As these applications only illustrate the variety of (potential) applications of our correlated sampling for the hypersimplex, we defer their description to \Cref{sec:applications}.
However, for our last application, we note that our dimension-free stretch hints at further applications: combining it with a simple rounding algorithm for the case that $k$ is large, we can obtain \emph{constant-stretch}, while approximately preserving marginals and submodular objectives. We believe this pattern will benefit future applications of our algorithm.

\section{Preliminaries}

\paragraph{Notation.} Throughout, we use boldface letters $\vec{x},\vec{y}$, etc.~to indicate vectors. As is standard, we use $\norm{\vec{x}}_1\triangleq \sum_i |x_i|$ to denote the $\ell_1$ norm. 
In our running times, we use the standard iterated (binary) logarithm notation. (Throughout the paper, all logarithms are to the base two.)
\begin{Def}\label{def:log*}
For any $n$, the $i$-th iterated and the iterated logarithm are defined as follows.
\begin{align*}
\log^{(i)}n&\triangleq \underbrace{\log \log \dots \log }_{i \textrm{ times}} n,\\
\log^*n &\triangleq  \min\{i \mid \log^{(i)}n\leq 1\}.
\end{align*}
\end{Def}

\paragraph{Sublinear sample time.} Our algorithms' sampling time is sublinear in the input's dimension. Specifically, ignoring a modest $\log^*n$ factor, their running time is independent of $n$ and only linear in the sparsity of the input vector $\vec{x}$, i.e., $nnz(\vec{x})\triangleq |\{i \mid x_i\neq 0\}|$. (As $\vec{x}$ will be clear from context, we simply write $nnz$.)
For this, we rely on a sparse representation that allows us to avoid reading the entire vector. 
In particular, we assume a natural representation allowing us to iterate through the non-zero coordinates of $\vec{x}$ in increasing order in time $nnz(\vec{x})$.
Another operation we require is $O(1)$-time computation of standard bit-wise operations, XOR, AND, and NOT (see \Cref{sec:icalp17}).

\paragraph{Interface.} Correlated sampling algorithms must use a common random seed for sampling, otherwise we may get different output sets in two different invocations of the sampling step for the same vector $\vec{x}$, violating Property \ref{prop:distance} for vectors $\vec{x}=\vec{y}$.
Fittingly, correlated sampling algorithms use a common \emph{initialization} step, where we draw $\calA\sim\calD$, setting this randomness, and \emph{sampling} steps, where we compute $\calA(\vec{x})$.
Note that in all our algorithms and the previous algorithms we make use of, the randomness of initialization can implicitly be ``spread'' over later sampling steps, drawing each random bit the first time it is needed. 
Therefore, all algorithms presented in this paper have constant initialization time, and so we avoid mentioning this, and only mention the sample time.

\vspace{-0.2cm}
\paragraph{Previous Correlated Sampling Algorithms.}
For $\Delta_n$, \cite{kleinberg2002approximation,holenstein2007parallel,ge2011geometric,buchbinder2018simplex,angel2019pairwise,liu2022simple} give $2$-stretch correlated sampling algorithms. For $k>1$, \cite{chen2017correlated} provided an $O(\log n)$-stretch correlated sampling algorithm for $\Delta_{n,k}$, though a full proof of the latter is unavailable online.\footnote{We thank Roy Schwartz for sharing a preprint of the full version of \cite{chen2017correlated} with us.} Moreover, since we are interested in input-sparsity sample time and fast parallel and dynamic implementations of these, in \Cref{sec:correlated-sampling-appendix} we provide full proofs of these algorithms, and implementations, with guarantees as in the following propositions. We also provide similar guarantees for $O(1)$-stretch correlated sampling in $\Delta_{n,1}$ more generally, by a simple reduction.
\begin{prop}\label{prop:k=1}
    There exists a $2$-stretch $O(nnz)$ sample time correlated sampling algorithm~for~$\Delta_{n}$. The algorithm can be implemented in $O(\log n)$ parallel depth and $O(n)$ work and using $O(\log (nnz)) = O(\log n)$~update~time in dynamic settings.
\end{prop}

 \begin{prop}\label{prop:ICALP17algo}
    There exists a correlated sampling algorithm for $\Delta_{n,k}$ for all $n,k\geq 1$ with sample time $O(nnz)$ and stretch $2\lceil\log n\rceil \leq 2\log n + 2$. The algorithm can be implemented in $O(\log n)$ parallel depth and $O(n)$ work and using $O(\log n)$~update~time in dynamic settings.
\end{prop}
Our input-sparsity implementation of \cite{chen2017correlated} relies on constant-time lowest-common ancestor function, using which we only process relevant nodes in the postorder traversal of \cite{chen2017correlated}. 

\vspace{-0.2cm}
\paragraph{Bounding Stretch.}
Throughout, in our analysis we fix two vectors $\vec{x},\vec{y}\in \Delta_{n,k}$, and prove that we satisfy all three desired properties, \ref{prop:k-uniform}-\ref{prop:distance}. 
The following well-known observation \cite{chen2017correlated} will simplify our discussion. As usual, $\vec{e}_i$ is the $i^{th}$ basis vector in the standard basis for $\mathbb{R}^n$. 

\begin{obs}\label{obs:triangle-ineq}
    By triangle inequality, a correlated rounding algorithm has stretch $\alpha$ if and only if this stretch of $\alpha$ holds for any two points $\vec{x},\;\vec{y}=\vec{x}+\veps\cdot \vec{e}_i\in \Delta_{n,k}$, for infinitesimally small $\varepsilon>0$.
\end{obs}
\section{Composed Correlated Sampling}\label{sec:correlated-sampling}

This section is dedicated to providing, analyzing and implementing our main correlated sampling algorithm for the hypersimplex, whose main guarantee we restate here.

\correlatedsampling*
We present our algorithm in \Cref{sec:the-algo}, showing that it satisfies Properties \ref{prop:k-uniform} and \ref{prop:marginals}, and analyze its stretch in \Cref{sec:analysis}. We discuss running time (in various computational models) in \Cref{sec:runtime}. We then prove the algorithm preservers submodular objectives in \Cref{sec:submodularity}.

\subsection{The Composed Correlated Sampling Algorithm}\label{sec:the-algo}

Our composes several correlated sampling algorithms and additional hashing and random thresholds. We precede its formal description and pseudocode (\Cref{alg:rounding-compression-init,alg:rounding-compression,alg:rounding-compression-project}) with a high-level layered description and motivation, interlaced with elements of its analysis.

\subsubsection{Overview}\label{sec:overview} 
\Cref{prop:ICALP17algo} gives a correlated sampling algorithm for $\Delta_{n,k}$ with stretch $O(\log n)$.
Thus, to obtain a smaller stretch, we aim to  ``project'' in a sense the points onto a smaller dimension $d\ll n$, so as to obtain a stretch $O(\log d) \ll \Theta(\log n)$.
Assuming the projection is in some sense invertible, this allows us to ``lift'' the same improved stretch back to the original (larger) dimension $n$.

\medskip 
\noindent\textbf{A first attempt.}
Consider a vector $\vec{x}$. To try and achieve the desired invertible projection, we hash coordinates into $d\ll n$ buckets. 
In the pseudocode, $d=m^3$, where these parameters will be clarified below.
Assuming (for now) that the sum of $x$-values hashed to the same bucket is not greater than one, we give this bucket's coordinate a weight equal to the sum of the $x$-values hashed to it. Note that the obtained vector $\hat{\vec{x}}\in \mathbb{R}^d$ has the same $\ell_1$ norm as $\vec{x}$. We now apply a correlated sampling algorithm to the vector $\hat{\vec{x}}$ and obtain a subset $S$ of $[d]$ of size $\|\vec{x}\|_1$ (up to floor or ceiling). 
We then obtain a subset of $[n]$ with the right marginals, for each bucket $b$ whose coordinate in $[d]$ is output in $S$, by picking a single coordinate $i$ hashed to bucket $b$ with probability $x_i/\sum_{j: h(j)=b}x_j$, using a simple correlated sampling algorithm for the case $k=1$ (\Cref{prop:k=1}).

\medskip
\noindent\textbf{A second attempt.}
Unfortunately, the assumption that the sum of $x$-values hashed to each bucket is at most one might be violated.
However, by simple probabilistic arguments, one can show that taking sufficiently large $d\geq \poly(k)$, then \emph{small-valued} coordinates (whose individual $x$-values are at most, say, $1/10$) hashed to a bucket have sum exceeding one with probability at most $1/\poly(d)$. 
We then hash \emph{large-valued} coordinates, of which there are $O(k)\ll d$, into their own buckets, which are unlikely to collide with buckets to which other coordinates were hashed. 
Assuming no collisions or heavy buckets, the above is our desired invertible projection.
In case of $1/\poly(d)$-probability problematic events (heavy bucket or collisions), we fall back on some $\alpha$-stretch algorithm. For now we think of $\alpha=O(\log n)$, though generally we can (and do) plug in other values of $\alpha$. By linearity of expectation, this fall back algorithm contributes $\alpha/\poly(d)$ to our algorithm's stretch, which is $O(1)$, provided $d\gg \alpha$. (This paragraph motivates our choice of $d=\poly(m)=\poly(k,\alpha)$.)

\medskip\noindent\textbf{The full algorithm.}
Finally, our definition of small/large coordinates and heavy buckets can result in using two different correlated sampling algorithms (in different lines) to determine the output for different inputs $\vec{x}$ and $\vec{y}$ which only differ marginally in a single coordinate, i.e., $\vec{y}=\vec{x}+\eps\cdot \vec{e}_i$. 
For example, this coordinate can be small for $\vec{x}$, but large for $\vec{y}$, or the change in this coordinate may result in one bucket being heavy when rounding $\vec{y}$, but not $\vec{x}$.
This is extremely problematic, since this can result in our outputs for $\vec{x}$ and $\vec{y}$ being computed using different correlated sampling algorithms, and so these outputs are potentially uncorrelated (different) output sets, despite $\vec{x}$ and $\vec{y}$ being quite close.
This results in stretch possibly as high as $O(k/\eps)$.
To overcome this issue, we randomize our thresholds for items being large or small, and for buckets being heavy.
As we show, this results in the extremely problematic event outlined above happening with probability at most $\eps/k$. Therefore, the contribution to the stretch here is again only $O(k/\eps)\cdot \eps/k = O(1)$, and the final stretch is only $O(1)+O(\log d) = O(\log k + \log \alpha)$.
Our final desired stretch of $O(\log k)$ follows by performing this construction recursively, allowing us to plug in increasingly smaller values of $\alpha$.

\vspace{-0.1cm}
\subsubsection{Algorithm Description}
Our recursive algorithm relies on the correlated sampling algorithms of \Cref{prop:k=1,prop:ICALP17algo}, and another correlated-sampling algorithm $\calA$ for $\Delta_{n,k}$ with stretch $\alpha$. This algorithm $\calA$ can be the algorithm from \Cref{prop:ICALP17algo}, in which case $\alpha = 2\lceil \log n \rceil$, or (as we use in general) another correlated-sampling algorithm---with a smaller $\alpha$---obtained by recursing this construction. 
(We later use this algorithm recursively so that in each level the stretch from $\calA$ decreases and the overall stretch achieved is reduced.)  At a high level, we rely on the algorithm for $k=1$ with hashing to compress and map the input point $\vec{x}\in \Delta_{n,k}$ to a point in a lower-dimensional polytope $\hat{\vec{x}}\in \Delta_{m^3,k}$ for some $m^3 \leq n$. 
If we successfully lowered the dimension of the point to round in an invertible manner, we then round this point using the \cite{chen2017correlated} correlated-sampling algorithm from \Cref{prop:ICALP17algo} for this lower dimension $m^3$.
When the compression does not succeed at some level of the recursion (due to hash collisions---this happens with small probability), we simply fall back to using algorithm $\calA$ on the original $n$-dimensional input vector.

Our basic algorithm's pseudocode is given in \Cref{alg:rounding-compression-init,alg:rounding-compression,alg:rounding-compression-project}.
On initialization (\Cref{alg:rounding-compression-init}) we initialize the randomness used later (in \Cref{alg:rounding-compression}) by the other correlated-sampling algorithms.
In addition, we randomly hash the coordinates into some large $m\triangleq \m$ many buckets.
We then hash the coordinates and buckets into $[m^3]$, and draw two random thresholds $\sigma,\tau\sim \Uni[0,1/10]$.
In \Cref{alg:rounding-compression}, we call a  coordinate $i\in [n]$ \emph{small} if $0<x_i\leq 1/10+\sigma$, else we call it \emph{large}.
If any two buckets or large coordinates are hashed to the same value in $[m^3]$ (a hash collision), or if the small coordinates in any one bucket have $x$-value summing to more than $1-
\tau$ (an unusually heavy bucket), we simply run Algorithm $\calA$ on the input vector $\vec{x}$. 
Otherwise, we use  \Cref{alg:rounding-compression-project} to ``compress'' the small coordinates in each bucket (whose sum is less than one), using the $k=1$ algorithm of \Cref{prop:k=1}: we obtain a single coordinate $i$ called the \emph{representative}, and have it ``absorb'' all the $x$-value of the small coordinates in its bucket---by taking on their $x$-values' sum and nullifying all small coordinates' $x$-values in its bucket (by resetting those to zero). 
Finally, using the (bijective) mapping from the remaining non-zero coordinates to $[m^3]$ (with representatives using their buckets' hash), we run the \cite{chen2017correlated} algorithm on a vector $\hat{\vec{x}}$ of (smaller) dimension $m^3$. We then invert the bijective mapping from the remaining non-zero coordinates in $\vec{x}$ to $[m^3]$ to convert the output subset of $[m^3]$ to a subset of $[n]$, and return this~set.
\vspace{-0.05cm}
\begin{algorithm}[H]
	\caption{Initialization \Comment{Used before all runs of \Cref{alg:rounding-compression}}}
    \label{alg:rounding-compression-init}
 \begin{algorithmic}[1]
    \State Set $m \gets \m$. 
    \color{black}
    \State Initialize Algorithm $\calA_1$ of \Cref{prop:k=1} for $\Delta_{n,1}$. \label{line:init-k=1}\Comment{Init randomness}
    \State Initialize Algorithm $\calA$ with a stretch of $\alpha$ for $\Delta_{n,k}$. \Comment{Init randomness}
    \State Initialize Algorithm $\calA_k$ of \Cref{prop:ICALP17algo} for $\Delta_{m^3,k}$. \label{line:init-k=k}\Comment{Init randomness}
    \For{$b=1,\dots,m$}\label{line:R-start}
    \State $\calB_b\gets \emptyset$.
    \State Draw $R_b \sim \Uni[m^3]$. \Comment{Target hash for small items of bucket $b$}
    \EndFor
    \For{$i=1,
    \dots,n$}
    \State Draw $C_i \sim \Uni[m^3]$.\Comment{Target hash for coordinate $i$} 
    \State Add $i$ to $\calB_b$ for $b\sim \Uni[m]$. \Comment{Partition coordinates among buckets u.a.r.}
    \EndFor
    \State \label{line:R-end}Draw $\sigma \sim\Uni[0,1/10]$ and $\tau \sim \Uni[0,1/10]$.
    \end{algorithmic}	
\end{algorithm}	
\vspace{-0.5cm}
\begin{algorithm}[H]
	\caption{Hypersimplex Correlated Sampling}
	\label{alg:rounding-compression}
 \begin{algorithmic}[1]
    \State $S\gets \{i \mid 0 < x_i \leq 1/10+\sigma \}$. \Comment{Non-zero small items' coordinates}
    \State $L\gets \{i \mid x_i  > 1/10+\sigma \}$. \Comment{Large items' coordinates}
    \If{$\exists h\in [m^3]$ such that $|\{i\in L \mid C_i=h\}\cup \{b \mid R_b =h\}|>1$}\label{line:collision} \Comment{Hashing collision}
    \State \textbf{Return} $\calA(\vec{x})$.\label{line:bad-output-collision}
    
    \EndIf
    \If{$\exists b\in [m]$ such that $\sum_{i\in S\cap \calB_b} x_i > 1 - \tau$}\label{line:heavy-bucket} \Comment{Bucket too heavy}
    
    \State \textbf{Return} $\calA(\vec{x})$.\label{line:bad-output-heavy-load}
    
    \EndIf 
    \State $(\hat{\vec{x}},\vec{REP})\gets \textsc{compress}(\vec{x},S)$. 
    \State $T\gets \calA_k(\hat{\vec{x}})$. \label{line:good-case} 
    \State \textbf{Return} $\{i \mid C_i\in T\}\cup \bigcup_{b: R_b\in T} REP_b $. \label{line:good-output}
    \end{algorithmic}
\end{algorithm}	

\vspace{-0.5cm}
\begin{algorithm}[H]
	\caption{\textsc{compress}$(\vec{x},S)$} 
	\label{alg:rounding-compression-project}
 \begin{algorithmic}[1]
    \For{$b=1,\dots,m$}\label{line:start} 
    \State $S_b\gets S\cap \calB_b$.  \Comment{Small items in bucket $b$}
    \If{$\sum_{i\in S_b}x_i > 0$}
    \State $REP_b \gets \calA_1\left(\frac{x_i}{\sum_{i\in S_b} x_i}\cdot \mathds{1}[i\in S_b]\right)$. 
    \label{line:KT}\label{line:compression-start} \Comment{Single representative for small items in $b$}
    \For{$i\in S_b$}
    \State 
    $x_i \gets \mathds{1}[i\in REP_b]\cdot \sum_{i\in S_b} x_i$. \Comment{Compress small items}
    \label{line:compression-end}
    \EndFor
    \EndIf
    \EndFor 
    \State $\hat{\vec{x}}\gets \vec{0}^{[m^3]}$.
    \For{$i=1,\dots,n$ with $x_i>0$} \label{line:mapping-start} \Comment{Map non-zero coordinates to smaller vector}
    \If{$i\in REP_b$ for (unique) $b$}
    \State $\hat{x}_{R_b}\gets x_i$.
    \Else
    \State $\hat{x}_{C_i} \gets x_i$.
    \EndIf 
    \EndFor
    \State \textbf{Return} $(\hat{\vec{x}},\vec{REP})$.
    \end{algorithmic}
\end{algorithm}	
\subsubsection{First Observations}

We note that \Cref{alg:rounding-compression}, which we denote henceforth by $ALG$, outputs a set of cardinality at most $k$, i.e., it satisfies Property \ref{prop:k-uniform}.
\begin{obs}\label{obs:cardinality}
    $|ALG(\vec{x})| \in \left\{\lfloor \|\vec{x}\|_1 \rfloor,\lceil \|\vec{x}\|_1 \rceil \right\}$ for all $\vec{x}\in \Delta_{n,k}$.    
\end{obs}
\begin{proof}
Follows by Property \ref{prop:k-uniform} of $\calA$, whose output $ALG$ may return in either Line \ref{line:bad-output-collision} or \ref{line:bad-output-heavy-load}, and by the same property of $\calA_k$, since if $ALG$ returns in Line \ref{line:good-output}, then since $\sum_i \hat{{x}}_i=\sum_i {{x}}_i$ and $|ALG(\vec{x})|=|\calA_k(\hat{\vec{x}})|$ in this case, we have $|ALG(\vec{x})|=|\calA_k(\hat{\vec{x}})| \in \left\{\lfloor \|\hat{\vec{x}}\|_1 \rfloor,\lceil \|\hat{\vec{x}}\|_1 \rceil \right\} = \left\{\lfloor \|\vec{x}\|_1 \rfloor,\lceil \|\vec{x}\|_1 \rceil \right\}$.
\end{proof}

Similarly, we note that \Cref{alg:rounding-compression}'s output contains each element $i\in [n]$ with probability $x_i$, i.e., it satisfies Property \ref{prop:marginals}. In fact, we prove a slightly stronger property.
Specifically, for $\calR$ the randomness used by \Cref{alg:rounding-compression-init} in Lines \ref{line:R-start}--\ref{line:R-end}, we show the following.

\begin{obs}\label{obs:conditional-marginlas}
    For any realization $r$ of $\calR$, element $i\in [n]$ and vector $\vec{x}$ 
    $$\Pr[i\in ALG(\vec{x})\mid \calR = r] = x_i.$$
    Consequently, by total probability, $\Pr[i\in ALG(\vec{x})] = x_i.$
\end{obs}
\begin{proof} 
    If $\calR=r$ implies that $ALG$ outputs a set in \Cref{line:bad-output-collision} or \ref{line:bad-output-heavy-load}, then $[ALG(\vec{x})\mid \calR=r]=\calA(\vec{x})$, and so by Property \ref{prop:marginals} of $\calA$, 
    $$\Pr[i\in ALG(\vec{x})\mid \calR=r] = \Pr[i\in \calA(\vec{x})] = x_i.$$ 
    If $\calR=r$ implies that $ALG$ outputs a set in \Cref{line:good-output} and that $i \notin S$, then  $\hat{x}_{C_i}=x_{i}$ and by Property \ref{prop:marginals} of $\calA_k$, $$\Pr[i\in ALG(\vec{x}) \mid \calR=r] = \Pr[C_i\in \calA_k(\hat{\vec{x}})  \mid \calR=r] = \hat{x}_{c_i}=x_i.$$
    Finally, if $\calR=r$ implies that $ALG$ outputs a set in \Cref{line:good-output} and that $i \in S$, then $i \in ALG(\vec x)$ iff both $i\in REP_b$ and $R_b\in \calA_k(\hat{\vec{x}})$ hold. 
   By Property \ref{prop:marginals} of $\calA_k$ and $\calA_1$, and since $\hat{\vec{x}}_{R_b}=\sum_{i\in S_b} x_i$ independently of $REP_b$ (and generally $\hat{\vec{x}}$ is independent of $R_b$), we have that
   \begin{align*}
       \Pr[i\in ALG(\vec{x}) ] & = \Pr[R_b \in \calA_k(\hat{\vec x})\mid \calR = r] \cdot \Pr[i \in REP_b \mid \calR = r] 
       = \hat{\vec{x}}_{R_b} \cdot \frac{x_i}{\sum_{i\in S_b}x_i} 
       = x_i. \qedhere
   \end{align*}
\end{proof}

    The crux of the analysis is in bounding the algorithm's stretch, i.e., quantifying Property \ref{prop:distance}.

    \subsection{Analyzing the Algorithm's Stretch}\label{sec:analysis}
    
    \paragraph{Analysis overview.} When rounding two vectors $\vec{x},\vec{y}\in \Delta_{n,k}$, we have three possible scenarios, named based on their intuitively increasingly poor conditional stretch.
    \begin{Def}
        The pair of runs of $ALG$ on vectors $\vec{x},\vec{y}\in \Delta_{n,k}$ (with the same randomness) is 
        \begin{enumerate}
            \item \textbf{\emph{Good} (G)} if $ALG$ outputs a set for both $\vec{x}$ and $\vec{y}$ in \Cref{line:good-output},
            \item \textbf{\emph{Bad} (B)} if $ALG$ outputs a set for both $\vec{x}$ and $\vec{y}$ in either \Cref{line:bad-output-collision} or \Cref{line:bad-output-heavy-load}, and 
            \item \textbf{\emph{Tragic} (T)} if $ALG$ outputs a set for $\vec{x}$ in \Cref{line:bad-output-collision} or \Cref{line:bad-output-heavy-load} and for $\vec{y}$ in \Cref{line:good-output}, or vice~versa.
        \end{enumerate}
    \end{Def}
    The following observation and subsequent lemma motivate the above terminology.
    \begin{obs}
        For vectors $\vec{x},\vec{y}\in \Delta_{n,k}$, the sets $X\triangleq ALG(\vec{x})$ and $Y\triangleq ALG(\vec{y})$ satisfy
        \begin{align*}
            \E[|X\oplus Y| \;\mid T] & \leq 2k, \\
            \E[|X\oplus Y| \;\mid B] & \leq \alpha\cdot \|\vec{x}-\vec{y}\|_1.
        \end{align*}
    \end{obs}
    \begin{proof}
        The bound for the tragic event is trivial, since $|X|,|Y|\leq k$ by \Cref{obs:cardinality}, and thus by triangle inequality, $|X\oplus Y|\leq |X|+|Y|\leq 2k$ always.
        Next, conditioned on the bad event, we have $X=\calA(\vec{x})$ and $Y=\calA(\vec{y})$, so the bound follows by Property \ref{prop:distance} of the $\alpha$-stretch correlated sampling algorithm~$\calA$.
    \end{proof}

    Note that the conditional stretch for the tragic event is unbounded, since it holds for arbitrarily small $\|\vec{x}-\vec{y}\|_1=\eps$, as opposed to the bounded stretch conditioned on a bad pair of runs. This justifies our choice of a more positive term for the latter event.
    We now justify the choice of best term for the good event ($G$): we  show that it contributes $O(\log m) = O(\log \alpha + \log k)$ to the stretch, which is $o(\alpha)$ unless $\calA$ is already an algorithm with our target stretch of $O(\log k)$. As we show later in \Cref{lem:prob-bad-tragic}, we have that $\Pr[G]=1-o(1)$, and so this also implies that good runs have the smallest expected stretch of all three types of runs.

    \begin{restatable}{lem}{goodstretch}
    \label{lem:good-stretch} 
    For vectors $\vec{x},\vec{y}\in \Delta_{n,k}$, the sets $X\triangleq ALG(\vec{x})$ and $Y\triangleq ALG(\vec{y})$ satisfy
        $$\E\left[|X \oplus Y| \mid G\right]\cdot \Pr[G] = O (\log k + \log\alpha)\cdot \norm{\vec{x}-\vec{y}}_1.$$
    \end{restatable}
    \begin{proof}[Proof (Sketch)]
        By \Cref{prop:ICALP17algo}, the output sets have distance roughly $O(\log m^3)\cdot \norm{\hat{\vec{x}}-\hat{\vec{y}}}_1=O(\log k+\log \alpha)\cdot \norm{\hat{\vec{x}}-\hat{\vec{y}}}_1$ conditioned on the fact that runs for $\vec x$ and $\vec y$ are good. 
        
        Therefore we wish to bound $\E[\norm{\hat{\vec{x}}-\hat{\vec{y}}}_1 \mid G]$ in terms of $\norm{{\vec{x}}-{\vec{y}}}_1=\eps$. In the case that $x_i$ and $y_i$ are both small, or are both large,  $\hat{\vec{x}}$ and $\hat{\vec{y}}$ differ by a single coordinate, hence $\norm{\hat{\vec{x}}-\hat{\vec{y}}}_1 = \norm{\vec{x}-\vec{y}}_1 = \eps$. However if $y_i$ is large but $x_i$ is small, then $\norm{\hat{\vec{x}}-\hat{\vec{y}}}_1 = 2x_i+\eps = O(1)$. As $\sigma$ is randomly chosen, this last event happens with probability $O(\eps)=O(\norm{{\vec{x}}-{\vec{y}}}_1)$, which then implies our desired bound by total expectation.
        See \Cref{appendix:correlated-sampling} for details.
    \end{proof}

    Given the preceding bounds on the expected contribution to the stretch of good runs, we wish to upper bound the probability of a pair of runs on vectors $\vec{x}$ and $\vec{y}$ being bad or tragic. 
    To this end, we first upper bound the probability of \Cref{alg:rounding-compression} running $\calA$ on the original input vector, by proving that the probability of a hash collision is low, conditioned on any realization of~$\sigma$.

    \begin{lem}\label{lem:collisions}
        For any $\vec{x}\in \Delta_{n,k}$ and possible realization $s$ of $\sigma$,
        $$\Pr\left[\exists  h \in [m^3] \textrm{ such that } | \{i \notin S \mid C_i = h, x_i>0\} \cup \{b \mid R_b = h\} | > 1 \;\middle\vert \; \sigma =s \right] \leq \frac{1}{m}.$$
    \end{lem}
    \begin{proof}
        Since $\sum_i {{x}}_i\leq k$, the number of elements $i\in L$ is at most $\frac{k}{1/10 + \sigma}\leq 10k$, while the number of buckets is $m$.
        The number of pairs of such is therefore at most $\binom{10k + m}{2}$. 
        Since hashing is done uniformly, the probability of any pair colliding is $\frac{1}{m^3}$. By union bound over all pairs and using $m = \m \geq 10k/(\sqrt{2} -1)$, and so $10k+m\leq \sqrt{2}m$, the probability of a hash collision is at most
        \begin{align*}
            \binom{10k + m}{2} \cdot \frac{1}{m^3} & \leq \frac{(10k+m)^2}{2m^3} \leq \frac{1}{m}.
            \qedhere
        \end{align*}
    \end{proof}

    Next, we show that under the same conditioning, the probability of any bucket $b$ being heavier than $\frac{4}{5}\leq 1-\tau$ is likewise small.

    \begin{restatable}{lem}{heavybucket}\label{lem:heavy-bucket}
        For any $\vec{x}\in \Delta_{n,k}$ and possible realization $s$ of $\sigma$,
        $$\Pr\left[\exists b \in [m] \textrm{ such that } \sum_{i\in S \cap \calB_b} x_i > \frac{4}{5} \;\middle\vert\; \sigma = s \right]\leq \frac{1}{m}.$$
    \end{restatable}
    \begin{proof}[Proof (Sketch)]
        We divide small items into \emph{tiny} items, having $x$ value at most some $\lambda = O(1/\log m)$, and \emph{little} items. 
        For some $\calB_b$ to have $x$ value of small items exceeding $4/5$ requires either the sum of $x$-values of tiny items in $\calB_b$ to exceed its expectation by some constant, or at least some constant number $C$ of coarse items to belong to $\calB_b$.
        Both events have probability polynomially small in $m$, the former by the choice of $\lambda$ and Chernoff bounds, and the latter by union bound over the (few) $C$-tuples of the (somewhat few) little items.
        Union bounding over these events and over all $b\in [m]$ then gives the lemma. See \Cref{appendix:correlated-sampling} for details.
    \end{proof}

Using the preceding lemmas, we can now upper bound the bad and tragic events' probabilities.
    \begin{restatable}{lem}{probbadtragic}\label{lem:prob-bad-tragic}
        The probability that a pair of runs of \Cref{alg:rounding-compression} on $\vec{x}$ and $\vec{y}$ is bad or tragic satisfy
        \begin{align*}
            \Pr[B] & \leq \frac{2}{m}, \\
            \Pr[T] & \leq \frac{30\veps}{m}.
        \end{align*}
        
    \end{restatable}
    \begin{proof}[Proof (Sketch)]
        The first bound follows by union bound and the previous two lemmas.
        In contrast, if a pair of runs is \emph{tragic}, then \Cref{alg:rounding-compression} terminates in \Cref{line:good-output} for one input, but terminates early (Lines \ref{line:bad-output-collision} or \ref{line:bad-output-heavy-load}) for the other, due to either a hash collision or a heavy bucket. 
        All (three) cases causing such an event require one of the $1/m$-probability events from the previous lemmas to occur, and the random thresholds for an item being large or for a bucket being full to fall in a range of width at most $\eps=\|\vec{x}-\vec{y}|_1$. But due to $\sigma,\tau\sim \Uni[0,1/10]$, any one of these three bad cases then occurs with probability $10\eps/m$. The upper bound on $\Pr[T]$ then follows by the union bound.    
See \Cref{appendix:correlated-sampling} for details.
    \end{proof}

With the above in place, we are now ready to bound the expected stretch of \Cref{alg:rounding-compression}.

\begin{thm}\label{thm:base-of-recurrence}
    For all $\vec{x},\vec{y}\in \Delta_{n,k}$, the sets $X\triangleq ALG(\vec{x})$ and $Y\triangleq ALG(\vec{y})$ satisfy
    \begin{align*}
        \E\left[|X\oplus Y|\right] =  O (\log k + \log \alpha) \cdot \norm{\vec{x}-\vec{y}}_1.
    \end{align*}
\end{thm}
\begin{proof}
    By \Cref{obs:triangle-ineq}, we can assume that $\norm{\vec x - \vec y}_1 = \veps$. 
    By total expectation over the good, bad and tragic events, using Lemmas \ref{lem:good-stretch} and \ref{lem:prob-bad-tragic}, we get:
\begin{align*}
    \E[|X\oplus Y|] 
    & = \E[|X\oplus Y| \mid G]\cdot \Pr[G] + \E[|X\oplus Y| \mid B]\cdot \Pr[B] + \E[|X\oplus Y| \mid T]\cdot \Pr[T] \\ 
    & \leq O (\log k + \log \alpha) \cdot \norm{\vec{x}-\vec{y}}_1 + \alpha \cdot \frac{2}{m}\cdot \norm{\vec{x}-\vec{y}}_1 + \frac{2k}{\veps} \cdot \frac{30\veps}{m} \cdot \norm{\vec{x}-\vec{y}}_1\\
    &\leq O(\log k + \log \alpha) \cdot \norm{\vec x - \vec y}_1,
\end{align*}
where the last inequality follows from our choice of $m$ satisfying $m\geq 2\alpha$ and $m\geq 60k$.
\end{proof}

\color{black}
Taking Algorithm $\calA$ in \Cref{alg:rounding-compression} to be the $\alpha = O(\log n)$-stretch algorithm of \Cref{prop:ICALP17algo}, \Cref{thm:base-of-recurrence} immediately implies an $O(\log k + \log \log n)$-stretch correlated sampling algorithm.
We improve on this bound and obtain an $O(\log k)$ stretch by recursing our construction.
We start by describing the construction.

\medskip\noindent\textbf{Our recursive construction.} For the base case, $\calA'_0$ is the algorithm of \Cref{prop:ICALP17algo}. For all $i\geq 1$ we let $\calA'_i$ be \Cref{alg:rounding-compression} with Algorithm $\calA$ in \Cref{line:bad-output-collision,line:bad-output-heavy-load} chosen to be $\calA'_{i-1}$.\footnote{\label{better-recurrence}We can also consider running $\calA'_{i-1}$ in \Cref{line:good-output}, but we focus on the simpler recurrence.}
We denote by $\alpha_i$ the stretch of algorithm $\calA'_i$. By \Cref{prop:ICALP17algo}, $\alpha_0 \leq 2\log  n+2$. On the other hand, by \Cref{thm:base-of-recurrence}, for general $i$ we have for some constant $C>1$ the recurrence 
\begin{equation*}
\label{eqn:recurrence}
    \alpha_{i+1}\leq  C \cdot (\log k + \log \alpha_i).
\end{equation*}

Using this recursive construction we obtain our claimed $O(\log k)$ stretch, as follows.

\correlatedsampling*
\begin{proof}
Let $p\triangleq  3C \cdot \log k + 10C^2$, and note that $10C^2/\log(10C^2) \geq 3C$ for all $C>1$. Consequently, since $x/\log x$ is increasing for $x\geq 10C^2 > e$, we have that if $\alpha_i\geq p \geq 10C^2$, then $\alpha_i/3 \geq C\log \alpha_i$.
Moreover, if $\alpha_i\geq p$, then trivially $\alpha_i/3 \geq p/3 \geq C\cdot \log k$.
Combining the above we have that if $\alpha_i \geq p$ then $2\alpha_i/3 \geq C\cdot(\log k + \log \alpha_i) \geq \alpha_{i+1} $.
Thus, since $\alpha_0 = O(\log n)$,  for $i\triangleq \log_{3/2}(\alpha_0)=O(\log \log n)$, Algorithm $\calA'_i$ has stretch $\alpha_i \leq p = O(\log k)$.
\end{proof}

By the preceding proof, to get an $O(\log k)$ stretch we only need a recursion depth of $O(\log \log n)$. In the following section we refine this bound and consider other computational aspects of implementations of \Cref{alg:rounding-compression} and our derived recursive correlated sampling algorithm.

\subsection{Computational Considerations}\label{sec:runtime}

In our preceding proof of \Cref{thm:correlated-sampling}, we only used that the stretch decreases by a constant factor until it is below some $O(\log k)$ term.
However, since the dependence on $n$ in the stretch of each algorithm $\calA'_i$ is asymptotically \emph{logarithmic} in the stretch of the preceding algorithm $\calA'_i$, we can show that the first iterations decrease the stretch significantly faster and so significantly fewer levels of recursion are necessary to attain stretch $O(\log k)$, which translates into speedups for our algorithm.

\begin{restatable}{lem}{logstar}\label{lem:first-recursive-algos}
     The $i$-th algorithm $\calA'_i$ for all $i\leq \log^*n-2$ has stretch $\alpha_i = O(i\cdot \log k + \log^{(i)}n).$
\end{restatable}
The proof, which is a simple inductive argument, is deferred to \Cref{appendix:correlated-sampling}. In the same appendix we show that after $\log^*n$ levels of recurrence, giving a stretch of $O(\log k\cdot \log^*n)$, by the halving argument from our previous proof of \Cref{thm:correlated-sampling}, we then decrease the stretch by a further $\log^*n$ factor in only $O(\log\log^*n)$ more levels of recursion, yielding the following.
\begin{restatable}{thm}{logstarthm}\label{thm:logstartheorem}
For some $i=\log^*n+O(\log\log^*n) = O(\log^*n)$, Algorithm $\calA'_i$ has stretch $O(\log k)$.
\end{restatable}
\begin{rem}
    One can show that the more involved recurrence mentioned in \Cref{better-recurrence} converges exponentially faster: it allows to double the number of iterated logarithms in each level, thus improving on \Cref{lem:first-recursive-algos}, yielding a stretch $\alpha_i = O(i\cdot \log k + \log^{(\mathbf{2^i})}n)$ for $i = O(\log \log^*n)$, and thus this number of recursive calls (asymptotically) suffices for a stretch of $O(\log k)$. 
    For simplicity of exposition, we omit the details.
\end{rem}

In the remainder of this section we discuss the computational ramifications of \Cref{thm:logstartheorem} in various computational models, the most immediate one being the classic centralized setting. In all the following theorems, we note that 

\begin{thm}[Near-Input-Sparsity Sample Time]
    There exists an $O(\log k)$-stretch correlated sampling algorithm for $\Delta_{n,k}$ with sample time $O(nnz\cdot \log^*n)$ in the worst case.
\end{thm}
\begin{proof}
    By \Cref{thm:logstartheorem}, we can focus on implementing $\calA'_i$ for some $i=O(\log^*n)$.
    Compression in each of the $O(\log^*n)$ levels of the recurrence requires $O(nnz)$ time, since we spend time proportional to each bucket's vector, by \Cref{prop:k=1}.
    The final call to the rounding algorithm of \Cref{prop:ICALP17algo} (for whichever dimension $d\leq n$ we end up using)
    takes a further $O(nnz)$ time. Finally, returning from the recurrence and lifting the output subset to a larger universe requires $O(nnz)$ time in each level, for a further $O(nnz\cdot \log^*n)$ time.
\end{proof}

Our recurrence depth together with implementations of the basic correlated sampling algorithms (\Cref{prop:k=1,prop:ICALP17algo}) also implies that our algorithm can be (i) parallelized (at the cost of increasing the sample time to be slightly superlinear in $n$), with depth near logarithmic. Similarly, and that it can (ii) be dynamized with the near-logarithmic same update time. (See \Cref{appendix:correlated-sampling}).
\begin{restatable}{thm}{parallel}\emph{(Parallel Implementation)\textbf{.}}
   There exists an $O(\log k)$-stretch correlated sampling algorithm for $\Delta_{n,k}$ with sampling time  $O(n\cdot \log^*n)$ and depth $O(\log n\cdot \log^*n)$.
\end{restatable}

\begin{restatable}{thm}{dynamic}\emph{(Dynamic Implementation)\textbf{.}}
   There exists an $O(\log k)$-stretch correlated sampling algorithm for $\Delta_{n,k}$ with update time  $O(\log n \cdot \log^*n)$.
\end{restatable}

\subsection{Submodular Dominance}\label{sec:submodularity}

Since correlated sampling algorithms preserve marginals (Property \ref{prop:marginals}), they natural preserve linear objectives. In this section we show that our rounding scheme also preserves \emph{subomdular} objectives, capturing the ubiquitous notion of diminishing returns. We recall some basic definitions.

\subsubsection{Submodularity: Technical Background}

\begin{Def}A set function $f:2^E\to \mathbb{R}$ is \emph{submodular} if it satisfies the \emph{diminishing returns} property, namely 
$$f(e \mid S) \geq f(e \mid T) \qquad \forall S\subseteq T\subseteq E\setminus \{e\},$$
where $f(e \mid S)\triangleq f(S\cup \{e\}) - f(S)$ denotes the \emph{marginal value} (measured by $f$) of adding $e$ to $S$.
\end{Def}

As we are interested in rounding-based algorithms, we need a way to extend (binary-valued) submodular functions to real vectors. We One natural such extension, corresponding to independent rounding, is the \emph{multilinear extension} of \cite{calinescu2011maximizing}. 
\begin{Def}\label{def:multilinear}
    The \emph{multilinear extension} $F:[0,1]^E\to \mathbb{R}$ of a set function $f:2^E\to \mathbb{R}$ is~given~by $$F(\vec{x})\triangleq \sum_{S\subseteq E} f(S) \prod_{i\in S} x_i \prod_{i\not\in S} (1-x_i).$$
\end{Def}

The maximum multilinear objective subject to any solvable polytope (i.e., a polytope over which one can optimize linear objectives efficiently) can be $(1-1/e)$-approximated in polytime \cite{calinescu2011maximizing}. This is optimal even subject to cardinality constraints, under standard complexity-theoretic assumptions \cite{feige1998threshold} and information-theoretically given only value oracle access to $f$ \cite{nemhauser1978best,mirrokni2008tight}.

Given the above approximability of the multilinear extension, it is natural to wish to round vectors $\vec{x}$ while preserving the constraints and preserving the multilinear objective. 
This is known to be achievable for matroid constraints and other constraints if $f$ is known \cite{calinescu2011maximizing,chekuri2010dependent}.
An \emph{objective-oblivious} counterpart to the above is given by the following definition.

\begin{Def}\label{def:submod-dom}
    A random vector $\vec{X}\in \{0,1\}^E$ satisfies \emph{submodular dominance} if for any submodular  function $f:2^E\to \mathbb{R}$, we have that
    $$\E[f(\vec{X})]\geq F(\E[\vec{X}]).$$
\end{Def}

It is known that tree-based pivotal sampling \cite{srinivasan2001distributions} (the core algorithm used by \cite{chen2017correlated}) satisfies submodular dominance.
This follows since tree-based pivotal sampling satisfies \emph{negative association} (and more) \cite{branden2012negative,dubhashi2007positive,dubhashi1996negative,byrka2018proportional}, and the latter is known to imply submodular dominance \cite{qiu2022submodular,christofides2004connection}.
We provide some relevant background on negative association here.

\begin{Def}
    A random vector $\vec{X}$ is \emph{negatively associated (NA)} if for any two increasing functions $f,g$ of disjoint variables in $\vec{X}$, we have $\Cov(f,g)\leq 0.$
\end{Def}
A simple example of the above is given by the so-called \emph{0-1 lemma} \cite{dubhashi1996balls}.
\begin{prop}\label{lem:0-1-lemma}
    Any random vector $\vec{X}\in \{0,1\}^n$ with $\sum_i X_i\leq 1$ is NA.
\end{prop}
It is well-known that NA is closed under products \cite{joag1983negative}.
\begin{prop}\label{lem:closure-products}
    If independent vectors $\vec{X}$ and $\vec{Y}$ are NA, then so is their concatenation, $\vec{X}\circ \vec{Y}$.
\end{prop}
A more involved example of NA is the output of \cite{chen2017correlated}. This follows since the latter is the output of pivotal sampling \cite{srinivasan2001distributions,deville1998unequal} (albeit with correlated random seed) with a tree order (see \Cref{appendix:correlated-sampling}), whose output is known to satisfy NA \cite{branden2012negative,byrka2018proportional,kramer2011negative}.
\begin{prop}\label{prop:NA-pivotal}
    The output of \cite{chen2017correlated} is NA.
\end{prop}

Recall that our interest in NA is due to its implication of submodular dominance \cite{christofides2004connection,qiu2022submodular}.
\vspace{-0.35cm}
\begin{prop}\label{lem:NA->submod}
    Let $f$ be a submodular function and $\vec{X}$ an NA vector. Then, $\E[f(\vec{X})]\geq F(\E[\vec{X}])$.
\end{prop}

\subsubsection{Submodular Dominance of \texorpdfstring{\Cref{alg:rounding-compression}}{Algorithm \ref{alg:rounding-compression}}}

In this section we show that the output of \Cref{alg:rounding-compression} satisfies submodular dominance, 
using the connection of the latter to negative association.

\begin{thm}\label{thm:correlated-sampling-with-submod}
    The output of \Cref{alg:rounding-compression} with $\calA$ and $\calA_k$ satisfying submodular dominance itself satisfies submodular dominance.
\end{thm}
\begin{proof}
Fix a realization $r$ of $\calR$. This in particular fixes the definition of large and small items, the buckets and the mappings.\todo{double check this references right lines}
If $r$ is such that we return a set in \Cref{line:bad-output-collision} or \Cref{line:bad-output-heavy-load}, then our output is $\calA_k(\vec{x})$, which satisfies submodular dominance by \Cref{prop:NA-pivotal}, and so by Property \ref{prop:marginals} of $\calA_k$, $$\E[f(ALG(\vec{x})) \mid \calR=r] = \E[f(\calA_k(\vec{x})) \mid \calR=r] \geq F(\E[\calA_k(\vec{x}) \mid \calR=r]) = F(\vec{x}).$$
Otherwise, given $T=\calA(\widehat{\vec{x}})$, our \Cref{alg:rounding-compression} outputs in \Cref{line:good-output} a set $$S(T)\triangleq \left\{i\mid C_i\in T\right\}\cup \bigcup_{b:R_b\in T} Rep_b.$$
To argue submodular dominance, we define the following auxiliary submodular function:
$$\widehat{f}(T) \triangleq \sum_{\substack{(i_1,\dots,i_r)\in (b_1,\dots,b_r):\\ T\cap \bigcup_b \{R_b\} =\{R_{b_1},\dots, R_{b_r}\}}} f\left((S(T)\cap L) \cup\bigcup_{j=1}^r \{i_j\}\right)\prod_{j=1}^r \frac{x_{i_j}}{\widehat{x}_{R_b}}.$$
Since the output of $\calA$ is independent of $Rep_b$, for any realization of $T$ we can imagine that (independently) for each $R_b\in T$, some $Rep_b$ is drawn from $\calB_b$ according to the discrete distribution $(x_i/\widehat{x}_{R_b})_{i\in \calB_b}$, by Property \ref{prop:marginals} of $\calA_0$.
Thus $\widehat{f}(T)=\E[f(S(T))]$ according to this sampling of $\{Rep_b\}_b$. So, if we let $I\triangleq \{(i_1,\dots,i_r)\in S^r \mid i_1,\dots,i_r \textrm{ belong to distinct $\calB_b$}\}$, then taking total probability over $T$, expanding $\widehat{f}(\cdot)$, and using the subomdular dominance of $\calA$,
we have:
\begin{align*}
    \E[f(S(T))] & =\E[\widehat{f}(T)] \\
    & \geq \sum_{T\subseteq [m^3]} \widehat{f}(T)\prod_{i\in T}x_i \prod_{i\in [m^3]\setminus T}(1-x_i) \\
    & = \sum_{S\subseteq L}\sum_{(i_1,\dots,i_r)\in I} f\left(L_s\cup \bigcup_{j=1}^r \{i_j\}\right)\prod_{i\in S}x_i \prod_{i\in L\setminus S}(1-x_i)\prod_{j=1}^r x_{i_j} \prod_{b: R_b\cap \bigcup_j \{i_j\}=\emptyset} (1-\widehat{x}_{R_b}).
\end{align*}
Fortunately, this unwieldy expression has a simple interpretation: this is precisely the subomdular objective obtained by sampling each large item $i$ independently with probability $x_i$, and sampling at most one element $i$ per bucket $\calB_b$ with probability $x_i$, independently of the large items and other buckets. But by \Cref{lem:0-1-lemma,lem:closure-products}, this distribution is NA, and since NA implies submodular dominance \cite{christofides2004connection,qiu2022submodular}, we obtain submodular dominance conditioned on $\calR=r$. By \Cref{obs:conditional-marginlas}, this implies that 
$\E[f(ALG(\vec{x})) \mid \calR=r] \geq F(ALG(\vec{x}))$.
The subomdular dominance for the unconditional distribution then follows since by total expectation over $\calR$.
\end{proof}
\section{Applications}\label{sec:applications}

Recall that correlated sampling algorithms for $\Delta_n$ have found numerous varied applications. In this section we discuss three applications of our rounding scheme: (i) an application to rounding fractional online paging algorithms, which while yielding worse algorithms than known, is simple, and hints at more applications in online settings, (ii) an application to metric multi-labeling via a framework of \cite{chen2017correlated}, improving their approximation from $O(\log n)$ to $O(\log k)$, and (iii) a more complicated application to collusion-free and swift submodular welfare maximizing reallocation in dynamic settings.
We conclude the latter with a discussion of combining dimension-free bounds with a simple rounding scheme to effectively obtain constant stretch, at the cost of only nearly achieving the marginals (Property \ref{prop:marginals}), and hence incurring a $(1-\eps)$ loss in submodular objectives.

We emphasize that our main message here is the applications' variety, which we believe hints at potential future wide-ranging applications of our more general correlated sampling machinery.






\subsection{Online Paging}
Recall that in the (online) paging problem, a cache of size $k$ is given, able to store any of $n$ pages.
At each time $t$ a page $i_t\in [n]$ is requested. If $i_t$ is not in the cache and the cache is full, this causes a \emph{cache miss} and some other page must be \emph{evicted} from the cache to make room for $i_t$.
The objective is to minimize the number of evictions (hence cache misses), compared to the hindsight-optimal solution, measured in terms of the (multiplicative) \emph{competitive ratio}.

A natural LP (linear programming) relaxation of the problem, with decision variables $y_{i,t}\in [0,1]$ corresponding to the extent to which page $i$ is in the cache at time $t$, and $z_{i,t}=|y_{i,t}-y_{i,t-1}|$ the extent to which page $i$ is evicted at time $t$, is as follows.
\begin{alignat*}{4}
    \min\; & \sum_i \sum_t z_{i,t} \\
    \mathrm{s.t.}\; & y_{i_t,t} \geq 1 && \forall t  \\
    & y_{i,t} - y_{i,t-1} \leq z_{i,t} &&\forall i,t \\
    & y_{i,t} - y_{i,t-1} \leq z_{i,t} & \quad &\forall i,t \\
    & y_{i,t}\geq 0 && \forall i,t.
\end{alignat*}

No fractional online algorithm is better than $O(\log k)$-competitive with respect to this LP, and this ratio can be attained using the online primal-dual method \cite{buchbinder2007online}.
We show that our correlated rounding algorithms provide $O(\log^2 k)$-competitive randomized integral algorithms. While better randomized algorithms with competitive ratio $O(\log k)$ have been known since the '90s \cite{fiat1991competitive}, this hints at broader applicability of our correlated sampling algorithm for online problems.

\paragraph{Rounding fractional caches using correlated sampling.} To round fractional caching algorithms is we use our $O(\log k)$-stretch correlated sampling \Cref{alg:rounding-compression} for $\calP_{n,k}$, abbreviated by $\calA$.
For $\vec{x}^{(t)}$  the fractional cache at time $t$, our integral cache at time $t$ is simply $\calA(\vec{x}^{(t)})$.
The random cache at each time $t$ is feasible, since (i) $|\calA(\vec{x}^{(t)})|\leq k$ by Property \ref{prop:k-uniform}, and (ii) $\Pr[i_t\in \calA(\vec{x}^{(t)})]=y_{i,t}$ by Property \ref{prop:marginals} implies that this cache contains page $i_t$ with probability one. 
On the other hand, the expected number of page misses at time $t$ is, by Property \ref{prop:distance}, at most $|\calA(\vec{x}^{(t)}) \oplus \calA(\vec{x}^{(t-1)})| = O(\log k)\cdot \norm{\vec{x}^{(t)}- \vec{x}^{(t-1)}}_1 = O(\log k)\cdot 
\sum_{i,t} z_{i,t}$. 
Therefore, by linearity of expectation, the obtained (feasible) online paging algorithm incurs at most $O(\log k)\cdot \sum_{i,t}z_{i,t}$ evictions in expectation. Therefore, since the fractional online algorithm's objective $\sum_{i,t}z_{i,t}\leq O(\log k)\cdot OPT $, the randomized algorithm is at most $O(\log^2 k)$ competitive.

\subsection{Metric Multi-Labeling}

Chen et al.~\cite{chen2017correlated} introduced the correlated sampling problem which we study. Their motivation comes from multi-label classification problems, in which labels need to be assigned to objects, e.g., news articles, given some observed data. They considered the case where multiple labels can be assigned to objects, called the {\em metric multi-labeling} problem. It
arises in various settings, e.g., classification of textual data such as web pages, semantic tagging of images and videos, and functional genomics. The assignment of labels to objects should be done in a manner that is most consistent with the observed data, from which two important ingredients are derived.
The first is an {\em assignment cost} for every (object,label) pair, reflecting a recommendation given by a local learning process which infers label preferences of objects.
The second is similarity information on pairs of objects, giving rise to {\em separation costs} incurred once different label sets are assigned to a pair of similar objects.
The goal is to find a labeling that minimizes a global cost function, while taking into account both local and pairwise information.

Chen et al. \cite{chen2017correlated} considered the setting in which the number of labels that an object can receive is at most $k\ll n$. They formulated the metric multi-labeling problem in this setting as a linear program that maps the $m$ objects into a set of (fractional) vectors $\by^1,\ldots,\by ^m\in \Delta_{n,k}$, where the $i^{th}$ entry of vector $j$ indicates the fraction of label $i$ that object $j$ receives. The objective function minimizes the global cost function of the (fractional) labeling, i.e., the sum of the assignment costs and the separation costs. 
The following lemma is proved in \cite{chen2017correlated}.
\begin{prop}
    If there is a polynomial-time $\alpha$-stretch correlated sampling algorithm for $\Delta_{n,k}$, then the metric multi-labeling problem admits a polynomial-time $\alpha$-approximation algorithm.
\end{prop}
The proof of the lemma is based on the following observations: (1) the preservation of the marginals in the correlated sampling algorithm guarantees that assignment costs are preserved in expectation; and (2) the $\alpha$-stretch implies that separation costs are preserved in expectation with a loss of a factor $\alpha$.
As discussed in detail above and in \Cref{appendix:correlated-sampling}, \cite{chen2017correlated} provided a polynomial-time $O(\log n)$-stretch algorithm, and from this they obtain an $O(\log n)$-approximate algorithm for metric multi-labeling.
Our main result that improves the stretch to $O(\log k)$ then implies the same improved approximation ratio for this problem. The latter is an asymptotic improvement of interest, as $k$ is typically much smaller than $n$ for this problem.

\subsection{Swift and Collusion-Free Reallocation}\label{sec:swf}

We now turn to our most involved application. Part of our arguments are inspired by and build upon recent unpublished work of \cite{buchbinder2025chasing}. 
They studied the problem of low-recourse cardinality-constrained submodular maximization in a dynamic setting with unknown future. 

In contrast, we show how our correlated sampling can be used to provide similar guarantees that are history-independent, and therefore collusion-resistant, given a sequence of potential scenarios that might require swift reallocation.
This application relies on all properties of our algorithm, including submodular dominance (\Cref{thm:correlated-sampling-with-submod}), and dimension-free stretch $O(\log k)$. 
We show that the latter property combines nicely with a simple constant-stretch algorithm which approximately preserves marginals and submodular objectives up to an additional $(1-\eps)$ factor for the regime that $k\geq \poly(1/\eps)$ (by simple concentration arguments): this combination results in $O(\log(1/\eps)$-stretch at the cost of a $(1-\eps)$ loss in the objective, by running the appropriate (near-)correlated sampling algorithm depending on whether $k\geq \poly(1/\eps)$ or not.
We turn to describing the problem we study in this section.

\paragraph{Submodular welfare maximization.} In the basic submodular welfare (SWF) maximization problem, we have $n$ buyers and $m$ sellers. Each seller $s$ sells some number $k_s$ of homogeneous items from some universe $U$. Each buyer $b$ in turn has a monotone submodular valuation function $f_b:U\to\mathbb{R}_+$ (i.e., $f_b(S)\leq f_b(T)$ for every $S\subseteq T\subseteq U$).
The objective is to allocate these items, at most one of each, to the buyers, so as to maximize the social welfare, $\sum_b f_b(S_b)$, for $S_b\subseteq U$ the set of items allocated to $b$.
We can encode the constraint that at most $k_s$ items sold by $s$ are allocated, at most one of which to each buyer, by requiring the (possibly fractional) allocation vector $\vec{x}$ with $(s,b)$ to be a concatenation of $m$ points in hypersimplexes, $\Delta_{n,k_s}$ with appropriate $k_s$.
\begin{alignat*}{4}
    & \sum_b x_{b,s} \leq k_s && \qquad \forall s \\
    & 0\leq x_{b,s}\leq 1 && \qquad \forall b,s.
\end{alignat*}

SWF is NP-hard to approximate within a factor $(1-1/e)$ \cite{khot2005inapproximability}.
There are many approaches to attain this approximation ratio in polynomial time. 
Most relevant to our subsequent dynamic SWF reallocation problem is the relax-and-round approach, as follows: the continuous greedy algorithm can be used to provide a $(1-1/e)$-approximation to the maximum multilinear extension subject to any solvable LP constraints, such as the above linear constraints \cite{calinescu2011maximizing}. 
One can then round the vectors using \Cref{alg:rounding-compression}, whose output satisfies the hypersimplex constraints (Property \ref{prop:k-uniform}, matches the marginals (Property \ref{prop:marginals}) and satisfies submodular dominance (see \Cref{thm:correlated-sampling-with-submod}) (see \Cref{prop:NA-pivotal}), applying this algorithm to each sub-vector $\vec{x}^s$ with coordinates $\{(s,b) \mid b\in [n]\}$, corresponding to the cardinality constraint of $k_s$ for items sold by $s$. By the algorithm's properties it satisfies all constraints, and by \Cref{lem:uniform-to-partition}, this also satisfies submodular dominance for the entire vector (with the target marginals), and so the obtained value is at least as high as the multilinear extension, which gives us a $(1-1/e)$-approximation \cite{calinescu2011maximizing}.

\textbf{The dynamic problem.} Suppose now that we have $r$ possible \emph{scenarios},  corresponding to some buyers or sellers leaving or entering the market.
Let $SWF(i)$ denote the optimal SWF of scenario $i$.
We wish to provide a good approximation of $SWF(i)$, while guaranteeing swift reallocation when switching from one scenario to the other. 
In particular, we wish to provide a $\beta$ approximation (to be chosen shortly) for each configuration, while minimizing the maximum number of items reallocated.
This can be captured by the following LP constraints, where $x_{b,s,i}$ is the fractional allocation to buyer $b$ of items of seller $s$ under scenario $i$:
\begin{alignat}{4}
    \min\; & R \tag{MMD-LP}\label{MMD-LP} \\
    \mathrm{s.t.}\; & \sum_{b,s} z_{b,s,i,j}\leq R && \quad\forall i,j \notag\\
    & x_{b,s,i} - x_{b,s,j} \leq z_{b,s,i,j} &&\quad \forall b,s,i,j \notag\\
    & x_{b,s,j} - x_{b,s,i} \leq z_{b,s,i,j} && \quad \forall b,s,i,j \notag\\
    & \sum_b x_{b,s} \leq k_s && \quad \forall s\in \textrm{ scenario $i$} \notag\\
    & \sum_b x_{b,s} \leq 0 && \quad \forall s\not\in \textrm{ scenario $i$} \notag\\
    & \vec{x}\geq \vec{0}. \notag
\end{alignat}

The above formulation is still missing the SWF-approximation. 
Using the ideas of \cite{buchbinder2025chasing} (as outlined in \Cref{app:bnw}), we can provide in polynomial-time a solution to the above constraint subject to the multilinear extension of the objective submodular function over $\vec{x}_{\mid i}\triangleq (x_{b,s,i})_{b,s}$ (the fractional assignment for scenario $i$) being a $(1-1/e-\eps)^2$-approximation of $SWF(i)$ for each $i$.
The question remains: how do we round all these vectors in a way to obtain bona fide \emph{integral} solutions with low movement cost? 
\cite{buchbinder2025chasing} provide an approach that results in $O(R)$ movement cost (i.e., a constant-approximation of the optimal movement cost). Unfortunately, their algorithm is history dependent, which may incentivize buyers and sellers to time their arrival and departure.
Our correlated sampling algorithm, in contrast, is \emph{history-independent}. Using our correlated sampling algorithm, which satisfies the hard cardinality constraints (Property \ref{prop:k-uniform}), the target marginals (Property \ref{prop:marginals}) and submodular dominance (\Cref{thm:correlated-sampling-with-submod}), combined with \Cref{lem:uniform-to-partition}, we obtain an algorithm with $(1-1/e-\eps)^2$-approximation, which by the algorithm's stretch uses movement cost $O(R\cdot \log k)$, i.e., $O(\log k)$ times that of an optimal algorithm.

\subsection{Discussion: From Dimension-Free to Constant}\label{sec:discussion} The above stretch can in fact be decreased to a \emph{constant} for any constant $\eps>0$, and specifically to $O(\log(1/\eps))$, albeit at the cost of a (negligible) $(1-\eps)$ deterioration in approximation, and losing the history independence. Since this is subsumed by the history-dependent $O(1)$-stretch rounding of \cite{buchbinder2025chasing}, we only outline the idea here briefly.

By sampling a uniform threshold for each coordinate, we can decide which elements to take to our solution (which item to sell to which buyer), using stretch $O(1)$: when $x_i\geq \tau_i$ for i.i.d., $\tau_i\sim \Uni[0,1]$, we add $i$ to the output. 
This yields the correct marginals independently, and so yields a set of expected value exactly equal to the multilinear extension, $F(\vec{x})$, and moreover results in a set of expected size $\|\vec{x}\|\leq k$. 
To make this a feasible set, we can scale down the value by a factor of $1+\Theta(\sqrt{k})$, thus obtaining a set of expected size $k-\Theta(\sqrt{k \ln k})$. By simple concentration arguments, the probability that this set exceeds size $k$ is polynomially small in $k$, and by submodularity, one can show that taking only a subset of the sampled elements if more than $k$ are sampled will only incur an expected additive loss of $(1/\poly(k))\cdot f(OPT)$. 
Combined with the $(1-\Theta(\sqrt{k \ln k}))$ multiplicative loss from scaling down, this is a $1-\eps$ multiplicative factor and $\eps\cdot f(OPT)$ additive factor, provided $k=\Omega(\log(1/\eps)1/\eps^2)$ is sufficiently large.
On the other hand, by only keeping a subset of cardinality $\min\{k,|S|\}$ of the sampled elements $S$, it is easy to obtain stretch $O(1)$.
Thus, we get a $(1-\eps)$ loss in objective (provided the target value is $\Omega(f(OPT))$) with $O(1)$ stretch if $k$ is larger than some $\poly(1/\eps)$. Alternatively, if $k$ is smaller than $\poly(1/\eps)$, then by applying this paper's main correlated sampling algorithm, we get no loss in the objective value, but stretch $O(\log k) = O(\log(1/\eps))$, i.e., constant for any constant $\eps>0$. 
In either case, we incur a stretch no worse than a $1-\eps$ deterioration in approximation quality, and a constant stretch of $O(\log (1/\eps))$.
We believe this pattern will find future applications.
\section{Summary and Open Questions}

In this work we revisit the correlated sampling question for the hypersimplex, $\Delta_{n,k}$, and provide dimension-free stretch guarantees for the latter. 
We provide stretch $O(
\log k)$ using a recursive algorithm of depth $O(\log^* n)$.

We note that any dimension-dependent stretch $f(n)\ll \log n$ should imply by similar approaches an $O(f(k))$ stretch with $O(f^*(n))$ levels of recurrence, where $f^*(n)\triangleq \min\{i\mid f^{(i)}(n)\leq 1\}$ is ``the iterated $f$'', defined using the $i^{th}$ iteration of $f$, namely $f^{(i)}(x)\triangleq f(f^{(i-1)}(x))\cdot \mathds{1}[i>0] + x\cdot \mathds{1}[i=0]$.
However, better than logarithmic dependence on $n$ is as yet unknown. Is this inherent, and is $\Theta(\log k)$ the optimal stretch given Properties \ref{prop:k-uniform} and \ref{prop:marginals}? We leave this as a tantalizing open question. 

We presented a number of applications, to online paging, metric multi-labeling and swift and history-independent dynamic reallocation for approximate submodular welfare maximization. 
While we do not see any of these applications as particularly compelling on their own, the variety of applications hints at further applications of our correlated sampling machinery. 
In particular, we discussed in \Cref{sec:discussion}, the dimension-independent stretch can, for some applications, be combined with simple rounding schemes to obtain \emph{constant} stretch (while only approximately preserving marginals).
We are optimistic that similarly to correlated sampling for the probability simplex, such (dimension-free) correlated sampling for the hypersimplex will find broader applications. We leave the search for such applications as a direction for future research.

\appendix
\section*{APPENDIX}
\section{Deferred Proofs of \texorpdfstring{\Cref{sec:correlated-sampling}}{Section \ref{sec:correlated-sampling}}}\label{appendix:correlated-sampling}

In this section we expand on the proofs only sketched in \Cref{sec:correlated-sampling}, restating the lemmas for ease of reference. We start by bounding the contribution to the stretch from a good pair of runs.

\goodstretch*
\begin{proof}
    By \Cref{obs:triangle-ineq}, we have $\vec{y} = \vec{x} + \eps \cdot \vec{e}_i$ for some $i\in [n]$, so that $\norm{\vec x - \vec y}_1 = \veps$. Let $\hat{\vec x}$ and $\hat{\vec y}$ be the vectors of dimension $m^3$ obtained after compressing and hashing the elements in $\vec x$ and $\vec y$ respectively. By \Cref{prop:ICALP17algo}, and since the mapping of the coordinates from $\vec x$ and $\vec y$ to $\hat{\vec x}$ and $\hat{\vec y}$ is the same in a good pair of runs (as the randomness initialized in \Cref{alg:rounding-compression-init} is shared), we have: 
    \begin{align}\label{eq: term1}
    \E\left[|T(\vec x)\oplus T(\vec y)| \mid G\right] \leq (2\log (m^3) + 2)\cdot \E\left[\norm{\hat{\vec x}-\hat{\vec y}}_1 \mid G\right].
    \end{align}
    To recover the final outputs $X$ and $Y$ from $T(\vec x)$ and $T(\vec y)$, we invert the projection $[n]\to[m^3]$: each large item is mapped back from $C_i\mapsto i$, and each bucket $b\in [m]$ is mapped from $R_b\mapsto \text{Rep}_b$ (where we let $\text{Rep}_b$ denote the single element representative drawn for that bucket). An additional discrepancy between $X$ and $Y$ may arise if, for the bucket $b$ containing $i$,   
    $R_b$ appears in both $T(\vec x)$ and $T(\vec y)$, but the respective representatives differ, i.e. $\text{Rep}_b(\vec{x}) \neq \text{Rep}_b(\vec{y})$. This increases the expected distance by at most 
\begin{align}\label{eq: term2}
    \E[|REP_b(\vec x) \oplus REP_b(\vec y)| \mid R_b\in (T(\vec x)\cap T(\vec y))\wedge G]\cdot\Pr[R_b\in (T(\vec x)\cap T(\vec y))\mid G].
\end{align}
We would thus get the following upper bound on the expected distance: $$\E\left[|X\oplus Y| \mid G\right]\leq (\ref{eq: term1})+
(\ref{eq: term2}),$$ and it remains to upper bound each of the two terms.

Before proceeding, observe that only the randomness generated in Lines \ref{line:R-start}---\ref{line:R-end} of \Cref{alg:rounding-compression-init} determines whether a pair of runs on the given input is good. This randomness also determines the projection of coordinates $[n]\to[m^3]$ and the partition of items into small and large ones. At the same time, randomness generated in Lines \ref{line:init-k=1} and Line \ref{line:init-k=k} for Algorithms $\calA_1$ and $\calA_k$ directly affects the distance between the outputs.

We will consider the following three disjoint cases:
\begin{itemize}
    \item[(i)] Event $E_1$: item $i$ is large in both runs, i.e., $\frac{1}{10} + \sigma < x_i<y_i$;
\item[(ii)] Event $E_2$: item $i$ is small in both runs, i.e., $x_i< y_i \leq \frac{1}{10}+\sigma$;
\item[(iii)] Event $E_3$: item $i$ is small in one run and large in the other, i.e., $x_i\leq \frac{1}{10}+\sigma< y_i$.
\end{itemize}

We claim that in cases (i) and (ii), $\norm{\hat{\vec x}-\hat{\vec y}}_1 = \norm{{\vec x}-{\vec y}}_1$ always holds.
Indeed, in the event $E_1$, the vectors $\hat{\vec x}$ and $\hat{\vec y}$ differ only at coordinate $C_i$: in particular, $\hat x_{C_i} = x_i$ and $\hat y_{C_i} = y_i = x_i+\eps$. Thus, $\norm{\hat{\vec x} - \hat{\vec y}}_1 = \veps$. In the event $E_2$, suppose bucket $b\in [m]$ is such that $i\in \calB_b$. Then $\hat{\vec x}$ and $\hat{\vec y}$ differ only at coordinate $R_b$, and $\hat y_{R_b} = \hat x_{R_b} +\veps= \sum_{j\in S_b(\vec x)}x_j + \veps$, which implies again that $\norm{\hat{\vec x} - \hat{\vec y}}_1 = \veps$.

As for the expression \ref{eq: term2}, since $\calA_1$ used in \Cref{line:KT} in \Cref{alg:rounding-compression-project} is a $2$-stretch algorithm, we obtain
    $$ (\ref{eq: term2})\leq
    2\cdot \E\left[\text{dist}_{l_1}\left(
    \frac{\vec x|_{S_b(\vec x)}}{\hat x_{R_b}},
    \frac{\vec y|_{S_b(\vec y)}}{\hat y_{R_b}}
    \right)\Bigg| R_b\in T(\vec x)\cap T(\vec y)\wedge G\right]\cdot \Pr[R_b\in T(\vec x)\cap T(\vec y)\mid G].
     $$
In the event $E_1$, we have that $\vec x|_{S_b(\vec x)} = \vec y|_{S_b(\vec x)}
$ and $\hat x_{R_b} = \hat y_{R_b}$
so the above expression equals $0$. Now, assume event $E_2$ occurs: then since $\frac{\vec x|_{S_b(\vec x)}}{\hat x_{R_b}}$ and $
    \frac{\vec y|_{S_b(\vec y)}}{\hat y_{R_b}}$ are probabilistic vectors, $x_i/\hat x_{R_b} < y_i/\hat y_{R_b} = (x_i+\eps)/(\hat x_{R_b}+\eps)$ and 
        $x_j/\hat x_{R_b} > y_j/\hat x_{R_b} = (x_j+\veps)/(\hat x_{R_b}+\veps)$ for $j\neq i$, and we obtain
that the hamming distance equals
$$
\text{dist}_{l_1}\left(
    \frac{\vec x|_{S_b(\vec x)}}{\hat x_{R_b}},
    \frac{\vec y|_{S_b(\vec y)}}{\hat y_{R_b}}
    \right)=
2\cdot \left(
\frac{x_i+\eps}{\hat x_{R_b}+\eps} - \frac{x_i}{\hat x_{R_b}}
\right)\leq \frac{2\eps}{\hat x_{R_b}}.
$$
On the other hand, by Property \ref{prop:marginals} we have $
\Pr[R_b\in (T(\vec x)\cap T(\vec y))\mid E_2\wedge G] = \hat x_{R_b}
$, thus, term (\ref{eq: term2}) is upper bounded by $4\cdot\norm{\vec x - \vec y}_1
$.

Finally, the case when event $E_3$ occurs is different, as in this case the vectors  $\hat{\vec x}$ and $\hat{\vec y}$ differ at two coordinates: $R_b$ and $C_i$. Specifically, $\hat y_{C_i} = y_i = x_i+\veps$ and $\hat x_{C_i} = 0$,  while $\hat y_{R_b} = \sum_{j\in S_b(\vec x)}y_j = \hat x_{R_b}+x_i$. Therefore, $\norm{\hat{\vec x} - \hat{\vec y}}_1 = 2x_i+\veps$, which is possibly much greater than $O(\eps)$. However, it is still $O(1)$ --- which as we shall soon see is sufficient since the probability that the event $E_3$ occurs is $O(\varepsilon)$. 
Indeed, since $\sigma\sim \textrm{Uni}[0,1/10]$, we have $\Pr[E_3] = \Pr[x_i -\frac{1}{10} \leq \sigma < y_i\frac{1}{10}] \leq 10\veps$. Following the same logic, it suffices to argue that term $(\ref{eq: term2})$ is upper bounded by $2$, simply because $|REP_b(\vec x)\oplus REP_b(\vec y)|\leq 2$ with probability $1$.

We are now ready to bound the expected $\ell_1$ distance in a good pair of runs, by total expectation:
\begin{align*}
    \E\left[\norm{\hat{\vec x} - \hat{\vec y}}_1 \mid G\right] &= \sum_{\ell=1}^3 \E\left[\norm{\hat{\vec x} - \hat{\vec y}}_1 \mid G \wedge E_\ell \right] \cdot \Pr[E_\ell\mid G] \\
        &\leq \norm{\vec x - \vec y}_1 \cdot(\Pr[E_1\mid G]+\Pr[E_2\mid G]) + (2x_i+\veps)\cdot\frac{\Pr[E_3]}{\Pr[G]} \\
        & \leq 
        \norm{\vec x - \vec y}_1 + 2\cdot\frac{10\eps }{\Pr[G]} & x_i\leq y_i = x_i+\eps\leq 1 \\
        & \leq \frac{21}{\Pr[G]}\cdot\norm{\vec x - \vec y}_1.  & \norm{\vec x - \vec y}_1 = \eps.
    \end{align*}

Similarly, we have
\begin{align*}
    (\ref{eq: term2})= \sum_{\ell =1}^3 
    &\E[|REP_b(\vec x) \oplus REP_b(\vec y)| \mid R_b\in (T(\vec x)\cap T(\vec y))\wedge E_\ell \wedge G]\\
    &\cdot\Pr[R_b\in (T(\vec x)\cap T(\vec y))\mid E_\ell \wedge G]\cdot \Pr[E_\ell\mid G]\\
    &\leq 4\norm{\vec x - \vec y}_1\cdot \Pr[E_2\mid G] + 2\cdot \frac{\Pr[E_3]}{\Pr[G]} \leq \frac{24}{\Pr[G]}\cdot \norm{\vec x - \vec y}_1.
\end{align*}

Combining the above and by \Cref{prop:ICALP17algo}, we obtain
\begin{align*}
    \E\!\left[|X \oplus Y| \mid G\right] \cdot \Pr[G] 
    &\leq \left((2\log (m^3) + 2)\cdot \frac{21}{\Pr[G]} + \frac{24}{\Pr[G]}\right)\cdot \|\vec{x} - \vec{y}\|_1 \cdot \Pr[G] 
    \\
    &\leq (42\log (m^3) + 66)\cdot \|\vec{x} - \vec{y}\|_1.
\end{align*}

Finally, by our choice of $m = \m$, and using $k,\alpha \geq 2$ so that $\log k, \log \alpha \geq 1$, we have
\begin{align*}
    \E\!\left[|X \oplus Y| \mid G\right] \cdot \Pr[G] 
    &\leq  (42\cdot 15(\log \alpha + \log k + 1) + 66)\cdot \|\vec{x} - \vec{y}\|_1 \\
    & = (42\cdot15(\log \alpha + \log k) + 348\cdot 2)\cdot \|\vec{x} - \vec{y}\|_1 \\
    & \leq (630\cdot(\log \alpha + \log k) + 348\cdot (\log \alpha + \log k))\cdot \|\vec{x} - \vec{y}\|_1 \\
    &= 978 \cdot (\log \alpha + \log k)\cdot \|\vec{x} - \vec{y}\|_1.
    \qedhere
\end{align*}
\end{proof}

We next prove that for any realization of $\sigma$, we are unlikely to have a particularly full bucket.
\heavybucket*
    \begin{proof}
     Fix $b\in [m]$. For simplicity of notation, we abuse notation and let $\Pr$ denote the probability space conditioned on $\sigma=s$. 
     We call a small item $j$ \emph{tiny} if $x_j$ is smaller than $\lambda \triangleq  \frac{1}{10\log  m + 5}$ and \emph{little} else.
     Let $S_b,T_b$ and $L_b$ denote the set of small, tiny and little items in bucket $\calB_b$, respectively. Then, by union bound
     \begin{align*}
         \Pr\left[\sum_{i\in S_b}x_i > \frac{4}{5} \right] 
            & \leq \Pr\left[\left(\sum_{i\in T_b}x_i > \frac{1}{5}\right) \bigvee_{} \left(\sum_{i\in L_b} x_i > \frac{3}{5}\right)\right] 
         \leq \Pr\left[\sum_{i\in T_b}x_i > \frac{1}{5}\right]+\Pr\left[\sum_{i\in L_b} x_i > \frac{3}{5}\right].
     \end{align*}
    
    As each tiny item is in $\calB_b$ with probability $\frac{1}{m}$, we have  $\mu\triangleq \E[\sum_{i\in T_b} x_i]\leq \E[\sum_{i\in \calB_b} x_i] \leq \frac{k}{m}$.
    Thus, $\frac{1}{5}\geq \frac{m}{5k}\mu \geq 2e\mu$, by our choice of $m$. 
    Therefore, since $Y_i\triangleq \mathds{1}[i\in T_b]\cdot  \frac{x_i}{\lambda} \in [0,1]$ are independent, by standard Chernoff bounds (see \cite[Theorem 4.4]{mitzenmacher2017probability}) we have that
    \begin{align*}       \Pr\left[\sum_{i\in T_b} x_i > \frac{1}{5}\right] &= \Pr\left[\sum_{i\in T_b} \frac{x_i}{\lambda}  > \frac{1}{5 \lambda} \right]  
        \leq 2^{-\frac{1}{5 \lambda}} = 2^{-2 \log m - 1} = \frac{1}{2m^2}.
    \end{align*}
    
Next, for the sum of little items (each of size at most $1/10+\sigma\leq 1/5$) in $\calB_b$ to exceed $3/5$,  this bucket must contain at least $4$ such items. There are at most $\|x\|/\lambda \leq k/\lambda$ of these, each placed in $\calB_b$ with probability $\frac{1}{m}$. 
Therefore, by union bound over all quadruples of little items, we have that
\begin{align*}
    \Pr\left[\sum_{i\in L_b} x_i > \frac{3}{5}\right] &\leq \binom{k/\lambda}{4}\frac{1}{m^4}  \leq \frac{k^4 (10\log m+5)^4}{24m^4}  \leq \frac{1}{2m^2}.
\end{align*}
Above, the last inequality follows from our choice of $m$ as follows: First, without loss of generality, $k\geq 2$, and $\alpha \geq 2$ by \cite{bavarian2020optimality}. Therefore, by our choice of $m$ we have $m = 2^{\constantC}k^{\expM}\alpha^{\expM} \geq 2^{11}k^4$, and $m\geq 2^{-11}12^{-1}(10\log m+5)^4$ (for all $m\geq 2^{15})$, and so $m^2\geq k^4(10\log m + 5)^4/12$, as desired above.
Finally, taking union bound over both events above over all $m$ buckets, the lemma follows.
\end{proof} 

We are now ready to upper bound the probability of bad and tragic pairs of runs.
\probbadtragic*
        \begin{proof}
        The bad event occurs when \Cref{alg:rounding-compression} terminates in Line \ref{line:bad-output-collision} or \ref{line:bad-output-heavy-load} for inputs $\vec{x}$ and $\vec{y}$ due to a hash collision or a heavy bucket. So, since $1 - \tau \geq \frac{4}{5}$, by  \Cref{lem:heavy-bucket,lem:collisions} and union bound, 
        \begin{align*}
            \Pr[\textrm{bad}] \leq \frac{1}{m}+\frac{1}{m}= \frac{2}{m}.
        \end{align*}
        
        The tragic event occurs if \Cref{alg:rounding-compression} terminates in \Cref{line:good-output} for one input, but terminates early in Lines \ref{line:bad-output-collision} or \ref{line:bad-output-heavy-load} for the other input as the result of a hash collision or a heavy bucket (i.e., the sum of small elements in some bucket is greater than $\frac{4}{5}$). There are three such possible situations, determined by $\calR$. We show that each occurs with probability at most $\frac{10 \veps}{m}$.

        Recall that $\vec{y} = \vec{x} + \eps\cdot \vec{e}_i$ for some $i\in [n]$, by \Cref{obs:triangle-ineq}. Let $S_b(\vec x)$ be the small coordinates in a bucket $b$ for $\vec x$ and define $S_b(\vec y)$ similarly. Let $E_1 \triangleq  \left[x_i \leq 1/10 + \sigma < x_i + \veps\right] $ be the event that coordinate $i$ is small for $\vec x$ but not for $\vec y$. Let $E_2 \triangleq  \left[\sum_{j \in S_b(\vec x)} x_j \leq 1 - \tau < \sum_{j \in S_b(\vec x)} x_j + \veps\right]$ be the event that the sum of small elements in bucket $b$ exceeds $1 - \tau$ for $\vec y$ but not for $\vec x$. As $\sigma$ and $\tau$ are distributed $\Uni\left[0, 1/{10}\right]$, the probability that their values fall within any fixed interval of width $\veps$ is at most $10 \veps$ (and strictly less if this interval is not contained in $[0,1/10]$). Consequently, $\Pr[E_1] \leq 10\veps$ and $\Pr[E_2] \leq 10\veps$.

        There are three cases, conditioning on events determined by the randomness fixed by $\calR$. For events $E_1$ and $E_2$, the probabilities are with respect to the randomness for choosing the target hashes, the choice of $\sigma$ and the choice of $\tau$. These events are independent of each other. 

        \textbf{Case 1:} The run for $\vec x$ ends in \Cref{line:bad-output-heavy-load} while the run for $\vec y$ ends in \Cref{line:good-output}. This occurs when  $i$ is small for $\vec{x}$ but not for $\vec{y}$ (i.e., $E_1$ holds) and there is a heavy bucket in $\vec x$, while there is no such bucket for $\vec y$.  Hence, by \Cref{lem:heavy-bucket}:
        \begin{align*}
            &\Pr\left[\exists b \textrm{ such that} \sum_{j \in S_b(\vec x)} x_j \geq \frac{4}{5} \textrm{ and } E_1\right] = \Pr\left[\exists b \textrm{ such that} \sum_{j \in S_b(\vec x)} x_j \geq \frac{4}{5} \;\middle\vert\; E_1\right] \cdot \Pr\left[E_1\right] \leq \frac{10\veps}{m}.
        \end{align*}
        
        \textbf{Case 2:} The run for $\vec y$ ends in \Cref{line:bad-output-collision} while the run for $\vec y$ ends in Line \Cref{line:good-output}. This occurs when $i$ is small for $\vec{x}$ but not for $\vec{y}$ (i.e., $E_1$ holds), and $C_i$ has a hash collision with a large coordinate or a bucket, i.e., $C_i=C_j$ for large item $j$ or $C_i=R_b$ for some bucket $\calB_b$.  Hence, by \Cref{lem:collisions}:
        \begin{align*}
            & \Pr\left[\exists j \textrm{ or } b \textrm{ such that } C_i = C_j \textrm{ or } C_i = R_b \textrm{ and } E_1\right] \\
            &= \Pr\left[\exists j \textrm{ or } b \textrm{ such that } C_i = C_j \textrm{ or } C_i = R_b \;\middle\vert\; E_1\right] \cdot \Pr[E_1] \\
            & \leq  \frac{1}{m} \cdot 10\veps = \frac{10\veps}{m}.
        \end{align*}

        \textbf{Case 3:} The run for $\vec x$ ends in \Cref{line:good-output} while the run for $\vec y$ ends in \Cref{line:bad-output-heavy-load}. This occurs when $i$ is small for $\vec{x}$ and $\vec{y}$, and there is a heavy bucket in $\vec y$, while there is no such bucket for $\vec x$. Note that $\sum_{j \in S_b(\vec x)} x_j + \veps = \sum_{j \in S_b(\vec y)} y_j$. Hence, by \Cref{lem:heavy-bucket}:
        \begin{align*}
            &\Pr\left[E_2 \textrm{ and } \exists b \textrm{ such that} \sum_{j \in S_b(\vec y)} y_j \geq \frac{4}{5} \right] \\
            &=  \Pr\left[E_2 \;\middle\vert\; \exists b \textrm{ such that} \sum_{j \in S_b(\vec y)} y_j \geq \frac{4}{5} \right] \cdot \Pr\left[\exists b \textrm{ such that} \sum_{j \in S_b(\vec y)} y_j \geq \frac{4}{5} \right]\\
            &\leq 10 \veps \cdot \frac{1}{m}  = \frac{10\veps}{m}.
        \end{align*}

        By union bound over the three events, the probability of the tragic event is at most $\frac{30\veps}{m}$. 
    \end{proof}

\paragraph{The recursive construction.} We now turn to upper bounding the number of levels needed by our algorithm to obtain $O(\log k)$ stretch. The key lemma is the following, proven by a simple inductive argument.
\logstar*
\begin{proof}
By \Cref{thm:base-of-recurrence}, the stretch $\alpha_i$ of Algorithm $\calA'_i$ satisfies $\alpha_i \leq C (\log k + \log\alpha_{i-1})$ for all $i\geq 1$, for some constant $C>4$.
Define $A \triangleq \log(Ck)$. We prove by induction on $i \in \{0,1,\dots,\log^*n-2\}$ that
\begin{equation*}\label{eq:bound}
\alpha_i \le C\left(\log^{(i+1)} n + iA\right).
\end{equation*}

The base case is immediate, since $\alpha_0 = 2\lceil\log n\rceil\leq 4\log n \leq C\log n$.
Similarly for $i=1$ we have $\alpha_i \leq C(\log k + \log\alpha_0) \leq C(\log k + \log C + \log \log n) = C(\log \log n + \log(Ck)) = C(\log^{(2)}n+A)$.
For the inductive step for $i\geq 2$, assume that the claimed bound holds for $i-1 \ge 1$, i.e.,
\[
\alpha_{i-1} \le C\left(\log^{(i)} n + (i-1)A\right).
\]
Using the recurrence $\alpha_i \le C(\log k + \log a_{i-1})$ and the definition of $A = \log(Ck)$, we obtain
\begin{align*}
\alpha_i 
&\le C\left(\log k + \log\bigl[C(\log^{(i)} n + (i-1)A)\bigr]\right) \\
&= C\left(A + \log\left(\log^{(i)} n + (i-1)A\right)\right) \\
& \leq C(\log^{(i+1)}n + iA),
\end{align*}
where the last inequality used
that for $x,y \ge 2$ one has $\log(x+y) \le \log(x\cdot y) = \log x + \log y \le \log x + y$; applying this identity to 
$x=\log^{(i)}n$ (which satisfies $x\geq 2$ for $i\leq \log^*n-2$ as we consider) and $y=(i-1)A\geq A = \log(Ck)\geq \log C \geq 2$,  yields the desired inequality, as
\begin{align*}
\log\left(\log^{(i)} n + (i-1)A\right) 
& \le \log^{(i+1)} n + (i-1)A. \qedhere 
\end{align*}
\end{proof}

\begin{cor}
    For $i=\log^*n-2$, Algorithm $\calA'_i$ has stretch $O(\log k \cdot \log^*n)$.
\end{cor}
Combined with our more heavy-handed proof strategy for \Cref{thm:correlated-sampling} in \Cref{sec:correlated-sampling}, we obtain the improved depth for our recurrence.
\logstarthm*
\begin{proof}
    The proof is essentially identical to the previous proof of \Cref{thm:correlated-sampling}, only starting with a better stretch $\alpha'_0 = O(\log k\cdot \log^*n)$ obtained from the preceding corollary using $\log^*n-2$ levels of recurrence. Consequently, by the same argument, after a further $\log_{3/2}(\alpha'_0 / \log k) = O(\log\log^*n)$ levels of recursion, we obtain a stretch $\alpha_i \leq p = O(\log k)$.
\end{proof}

Finally, we discuss simple parallel and dynamic byproducts of our recursion's modest depth.

\parallel*
\begin{proof}
    The basic building blocks of \Cref{alg:rounding-compression} consist of operations easily implemented in constant parallel depth (hashing and checking for collisions) or logarithmic depth (checking for overfull buckets), and applications of a correlated sampling for the probability simplex, which can be implemented in depth $O(\log n)$ and work $O(n)$, by \Cref{prop:k=1}.
    These operations are repeated for at most $O(\log^*n)$ levels (sequentially), with a final application of the algorithm of \Cref{prop:ICALP17algo} (for some dimension), which can also be implemented in depth $O(\log n)$ and work $O(n)$, by \Cref{prop:k=1}.
    Combining the above, the claimed work and depth follow.
\end{proof}

\dynamic*
\begin{proof}[Proof (Sketch)]
    The arguments resemble those of the parallel implementation, with the following modifications: note that if we maintain the output of each call to $\calA$ and $\calA_k$ (even when we do not output these in the algorithm), then since the basic correlated sampling algorithms for $k=1$ and $k>1$ which we use can be updated in time $O(\log n)$ (by \Cref{prop:k=1,prop:ICALP17algo}), then the algorithm's bottleneck is the maintenance of $O(\log^* n)$ such algorithms' outputs over all levels of the recursive construction. 
    The only delicate point is that \Cref{line:KT}, at face value, requires the scaling of several values per change to $\vec{x}$. However, inspecting \Cref{alg:exp-round}, we note that it is invariant under scaling of all its input's coordinates, and so we can in fact avoid the scaling altogether.
\end{proof}
\section{Previous Correlated Sampling Algorithms}\label{sec:correlated-sampling-appendix}

In this section we provide details of previous correlated sampling algorithms, for completeness.

\subsection{Correlated Sampling for \texorpdfstring{$\Delta_{n,1}$}{}}\label{app:k=1}

Two algorithms with optimal stretch of $2$ are known for rounding in the probability simplex, $\Delta_n$. These include an application of von-Neumann's rejection sampling to discrete settings, e.g., \cite{kleinberg2002approximation,holenstein2007parallel,liu2022simple}, and so-called exponential clocks, e.g., \cite{ge2011geometric,buchbinder2018simplex}. Since our goal is to show that these are also fast, we provide a self-contained analysis of the exponential clock algorithm.\todo{call it sticky randomness too?}
We then show how any such $2$-stretch algorithm for $\Delta_n$ can be extended simply (by adding a dummy coordinate) to obtain a $2$-stretch algorithm for the more general polytope $\Delta_{n,1}$, containing \textbf{sub-}distributions $\vec{x}\in [0,1]$ with $\sum_i x_i\mathbf{\leq} 1$. 

We start by describing and analyzing the exponential clocks-based algorithm.

\begin{algorithm}[H]
	\caption{Exponential Clocks for $\Delta_{n,1}$: Initialization}
	\label{alg:exp-init}
 \begin{algorithmic}[1]
    \State Sample i.i.d.~$R_i\sim \Exp(1)$ for all $i\in [n]$.
\end{algorithmic}	
\end{algorithm}	
\begin{algorithm}[H]
	\caption{Exponential Clocks for $\Delta_{n,1}$: Round}
	\label{alg:exp-round}
 \begin{algorithmic}[1]
    \State Output $\arg\min_i R_i/x_i$.
\end{algorithmic}	
\end{algorithm}	

\begin{prop}\label{prop:fast-exponential-clocks}
    \Cref{alg:exp-round} is a $2$-stretch algorithm for the probability simplex, with sample time $O(nnz)$, implementable in $O(\log n)$ depth and $O(\log (nnz)) = O(\log n)$~update~time.
\end{prop}
\begin{proof}
    Since sampling from exponential distributions can be done in constant time (assuming the same is true for sampling from $\Uni[0,1]$), the sampling time is immediate. Parallel depth follows by time to compute the minimum of $k$ values. Finally, update time can be obtained by maintaining a heap or balanced binary search tree of the values $R_i/x_i$ for non-zero $x_i$, in time logarithmic in their number, $nnz$.
    
    Next, the output set has cardinality one with probability one, so Property \ref{prop:k-uniform} is satisfied. Moreover, since $R_i\sim \Exp(1)$ implies $R_i/c\sim \Exp(c)$, and since $\Pr[Y_i \leq \min_{j\neq i} Y_j] = {\lambda_i}/{\sum_j \lambda_j}$ for independent $Y_1\sim \Exp(\lambda_1),\dots,Y_n\sim \Exp(\lambda_n)$, and since $\norm{\vec{x}}_1=1$ for input $\vec{x}$ in the probability simplex, $\Pr[ALG(\vec{x})=\{i\}] = \frac{x_i}{\norm{\vec{x}}_1}=x_i$ for all $i\in [n]$, i.e., Property \ref{prop:marginals} is satisfied.
    To prove that Property \ref{prop:distance} is satisfied with $\alpha=2$, by a similar argument to \Cref{obs:triangle-ineq} we may focus on two input vectors $\vec{y},\vec{x}$ of the form $\vec{y}=\vec{x}+\epsilon\cdot \vec{e}_1 - \epsilon\cdot \vec{e}_2$  (up to renaming of coordinates) for infinitesimally small $\epsilon>0$.
    To bound the expected deviation, it suffices to upper bound the probability that the sets differ, or equivalently to lower bound the probability that the output sets agree. Following the argument in \cite{buchbinder2018simplex}, we denote by $E_i$ the event that \Cref{alg:exp-round} outputs the singleton set $\{i\}$ when run both on $\vec{x}$ and $\vec{y}$, i.e., that $j=i$ minimizes both $Z_j/x_j$ and $Z_j/y_j$. 
    \begin{align*}
        \forall i>2: \qquad \Pr[E_i] & = \Pr\left[\frac{R_i}{x_i}\leq \min\left\{\frac{R_1}{x_1+\epsilon}, \frac{R_2}{x_2},\frac{R_3}{x_3},\cdots \frac{R_n}{x_n} \right\}\right] = \frac{x_i}{\norm{\vec{x}}_1+\epsilon} = \frac{x_i}{1+\epsilon} \geq x_i(1-\epsilon).\\
        \Pr[E_2] & = \Pr\left[\frac{R_2}{x_2-\epsilon}\leq \min\left\{\frac{R_1}{x_1+\epsilon}, \frac{R_3}{x_3},\cdots \frac{R_n}{x_n} \right\}\right] = \frac{x_2-\epsilon}{\norm{\vec{x}}_1} = x_2-\epsilon. \\
        \Pr[E_1] & = \Pr\left[\frac{R_1}{x_1}\leq \min\left\{\frac{R_2}{x_2}, \frac{R_3}{x_3},\cdots \frac{R_n}{x_n} \right\}\right] = \frac{x_1}{\norm{\vec{x}}_1} = x_1.
    \end{align*}
    Summing the above, we have that
    $$\Pr[ALG(\vec{x}) = ALG(\vec{y})]= \sum_i \Pr[E_i] \geq x_1 + (x_2-\epsilon) + \sum_{i>2} x_i(1-\epsilon) \geq 1-2\epsilon.$$
    Therefore, 
    \begin{align*}
    \E[|ALG(\vec{x})\Delta ALG(\vec{y})] & = 2\cdot (1-\Pr[ALG(\vec{x}) = ALG(\vec{y})]) \geq 4\epsilon = 2\cdot \norm{\vec{x}-\vec{y}}_1. \qedhere
    \end{align*}
\end{proof}

We next extend the above algorithm's guarantees to obtain fast correlated sampling for $\Delta_{n,1}$.
The simplest way to generalize from the probability simplex to the unit simplex (i.e., allowing for vectors with sum \emph{at most} one) is obtained by adding a dummy coordinate and projecting the result back to the original $n$ coordinates, as in \Cref{alg:KT-dummy-round}.
\begin{algorithm}[H]
	\caption{General $\Delta_{n,1}$ dummy coordinate rounding init}
	\label{alg:KT-dummy-init}
 \begin{algorithmic}[1]
    \State Initialize $2$-stretch algorithm $\calA$ for the $n$-dimensional probability simplex.
    \State Sample $R \sim \Uni[0,1]$.
\end{algorithmic}	
\end{algorithm}	
\vspace{-0.5cm}
\begin{algorithm}[H]
	\caption{General $\Delta_{n,1}$ dummy coordinate rounding}
	\label{alg:KT-dummy-round}
 \begin{algorithmic}[1]
    \State $\hat{\vec{x}}\gets \vec{x}+(1-\norm{\vec{x}}_1)\cdot \vec{e}_{n+1}$. \label{line:dummy-coord}
    \State $S\gets \calA(\hat{\vec{x}})$.
    \If{$S=[n+1]$}
    \State \textbf{Return} $\emptyset$.
    \Else 
    \State \textbf{Return} $S$.
    \EndIf
\end{algorithmic}	
\end{algorithm}	

\begin{lem}
    \Cref{alg:KT-dummy-round}, denoted by $\calB$, is a $3$-stretch correlated sampling algorithm for $\Delta_{n,1}$ implementable using $O(nnz)$ sampling time, $O(\log n)$ depth and $O(\log (nnz))=O(\log n)$ update~time. 
\end{lem}
\begin{proof}
    The implementation times are immediate, by \Cref{prop:fast-exponential-clocks}. Properties \ref{prop:k-uniform} and \ref{prop:marginals} are likewise immediate. Next, fix two vectors $\vec{x},\vec{y}\in \Delta_{n,1}$, which by \Cref{obs:triangle-ineq}, we can assume satisfy $\vec{y}=\vec{x}+\epsilon\cdot \vec{e}_i$ for some $i\in [n]$. Next, let $\hat{\vec{x}},\hat{\vec{y}}\in \Delta_{n+1,1}$ be the corresponding $(n+1)$-dimensional points defined in \Cref{line:dummy-coord} when run on $\vec{x}$ and $\vec{y}$, respectively.
    Then, $$\norm{\hat{\vec{x}}-\hat{\vec{y}}}_1 = \norm{{\vec{x}}-{\vec{y}}}_1 + |(1-\norm{\vec{x}}_1) - (1-\norm{\vec{y}}_1)| = 2\cdot \norm{{\vec{x}}-{\vec{y}}}_1 = 2\epsilon.$$
    Therefore, by Property \ref{prop:distance} of $\calA$, we have that $\E[|\calA(\hat{\vec{x}})\Delta \calA(\hat{\vec{y}})|]\leq 4\epsilon$.
    Accounting for the probability that $\calA(\hat{\vec{x}})=[n+1]\neq \calA(\hat{\vec{y}})$ (and vice versa), which does not contribute as much to the symmetric difference, we have by Property \ref{prop:marginals} of $\calA$, that
    \begin{align*}
        \E[|\calB(\vec{x})\Delta \calB(\vec{y})|] & = \E[|\calA(\hat{\vec{x}})\Delta \calA(\hat{\vec{y}})|] - \Pr[\calB({\vec{x}}) = \emptyset \wedge \calB({\vec{y}})\neq \emptyset] - \Pr[\calB({\vec{x}}) \neq \emptyset \wedge \calB({\vec{y}}) = \emptyset] \\
        & \leq \E[|\calA(\hat{\vec{x}})\Delta \calA(\hat{\vec{y}})|] - \Pr[\calB({\vec{x}}) \neq \emptyset \wedge \calB({\vec{y}}) = \emptyset] \\
        & = \E[|\calA(\hat{\vec{x}})\Delta \calA(\hat{\vec{y}})|] - \left(\Pr[\calB({\vec{y}}) = \emptyset] - \Pr[\calB({\vec{x}}) = \emptyset \wedge \calB({\vec{y}}) = \emptyset]\right) \\
        & \leq \E[|\calA(\hat{\vec{x}})\Delta \calA(\hat{\vec{y}})|] - \left(\Pr[\calB({\vec{y}}) = \emptyset] - \Pr[\calB({\vec{x}}) = \emptyset]\right) \\
        & \leq 4\epsilon - \epsilon = 3\epsilon. \qedhere 
    \end{align*}
\end{proof}

\subsection{Correlated Sampling for \texorpdfstring{$\Delta_{n,k}$}{} with \texorpdfstring{$O(\log n)$}{} stretch}\label{sec:icalp17}

In this section we provide a self-contained description and analysis of the stretch of the $O(\log n)$-stretch correlated sampling algorithm for $\Delta_{n,k}$ of \cite{chen2017correlated} (\Cref{prop:ICALP17algo}).

\paragraph{Algorithm description.}
Define \emph{coordinate tree} $T$ to be a full binary tree with $n$ leaves, one for each possible coordinate. The algorithm is correlated, that is the same random bits are used in the execution for all inputs. Draw random thresholds uniformly from $[0, 1]$ for all non-leaf nodes in $T$, and draw an additional threshold for the root $r$. The leaves of the tree pass their coordinate and marginal values to their parent. At each non-leaf node $z$, given the coordinates and marginal values from its children, we execute the rounding procedure given in \Cref{alg:node-rounding}. The random threshold is used to `fix' one of the coordinates, that is to round the corresponding marginal to $0$ or $1$. Both marginals are then updated,  preserving their sum, and the non-fixed coordinate and updated marginals are passed up to the parent of $z$.  At the root node after the above rounding procedure, the final coordinate is fixed using the threshold drawn for the root.
See pseudocode in \Cref{alg:tree-rounding}.

\begin{algorithm}[H]
	\caption{Initialization for Coordinate Tree Rounding($T$)}
	\label{alg:label-tree-initialization}
 \begin{algorithmic}[1]
    \State For every non-leaf node $z$, draw $\theta_z \sim \Uni[0, 1]$. 
    \State For root $r$ of $T$, draw $\theta_r \sim \Uni[0, 1]$.
    \end{algorithmic}	
\end{algorithm}	

\vspace{-0.5cm}
\begin{algorithm}[H]
	\caption{Resolve$(i, \alpha, j, \beta, \theta_z)$ \Comment{Executed at every non-leaf node}}
	\label{alg:node-rounding}
 \begin{algorithmic}[1]
    \If {$0 \leq \alpha + \beta \leq 1$} 
        \If {$\theta \leq \frac{\alpha}{\alpha + \beta}$}\label{alg:first-case}
        
        $\alpha' \leftarrow \alpha + \beta, \beta' \leftarrow 0, s \leftarrow j$.
        \Else\label{alg:second-case}
        
        $\alpha' \leftarrow 0, \beta' \leftarrow \alpha + \beta, s \leftarrow i$.
        \EndIf
    \EndIf

    \If {$1 \leq \alpha + \beta \leq 2$} 
        \If {$\theta \leq \frac{1- \beta}{2 - \alpha - \beta}$}
        
        $\alpha' \leftarrow 1, \beta' \leftarrow \alpha + \beta -1, s \leftarrow i$.
        \Else 
        
        $\alpha' \leftarrow \alpha + \beta - 1, \beta' \leftarrow 1, s \leftarrow j$.
        \EndIf
    \EndIf
    \State \textbf{Return}  {$(i, \alpha', j, \beta')$ and fix $s$}.
    \end{algorithmic}	
\end{algorithm}	

\vspace{-0.5cm}

\begin{algorithm}[H]
	\caption{Coordinate Tree Rounding$(\vec x, T)$ \Comment{Coordinate Tree Rounding}}
	\label{alg:tree-rounding}
 \begin{algorithmic}[1]
    \For{every leaf $i \in \{1, \dots, n\}$}
        \State Sends its parent coordinate $i$ and marginal value $x_i$.
    \EndFor
    \For{every non-leaf node $z \in V$, from leaf} 
        \State $z$ receives coordinates and marginals $(i, \alpha), (j, \beta)$ from its children. 
        \State Run \Cref{alg:node-rounding} on input $(i, \alpha, j, \beta, \theta_z)$ to obtain updated marginals $\alpha', \beta'$. 
        \State Send $\alpha', \beta'$, the non-fixed coordinate and updated marginals to the parent of $z$. 
    \EndFor
    \State Let $s$ be the last label that was not fixed with updated marginal $\gamma$. 
    \State Root $r$ rounds the marginal to 1 if $\gamma \leq \theta_r$ and to 0 otherwise.
    \State \textbf{Return}  {$\hat{\vec x}$ where $\hat{\vec x}_i = 1$ if coordinate $i$'s marginal was rounded to 1, and 0 otherwise}. 
    \end{algorithmic}	
\end{algorithm}	

As any tree on $n$ leaves has $O(n)$ nodes, and each node in tree $T$ requires $O(1)$ time, we have:
\begin{obs}
    \Cref{alg:tree-rounding} and \Cref{alg:label-tree-initialization} take $O(n)$ time.
\end{obs}

By a simple inductive argument
\Cref{alg:tree-rounding}
preserves the marginals and sum of marginals at each step, and hence satisfies the cardinality constraint  and marginal preservation  \cite{srinivasan2001distributions}.
\begin{prop}\label{prop:tree-properties}
\Cref{alg:tree-rounding} satisfies Properties \ref{prop:k-uniform} and \ref{prop:marginals}.
\end{prop}

\Cref{alg:tree-rounding} (with initialization \Cref{alg:label-tree-initialization}) was shown to have a stretch of $O(\log n)$ by \cite{chen2017correlated}. As a full proof is not available online, we provide one here for completeness.\footnote{Again, we thank Roy Schwartz for providing us a full version of \cite{chen2017correlated}.}

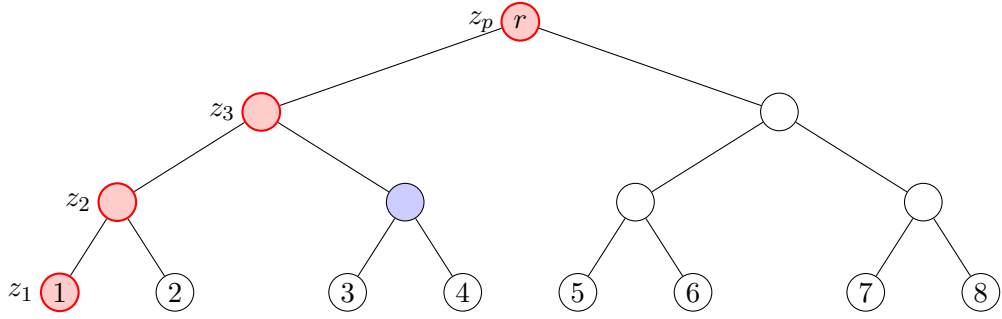
\begin{figure}
        \centering
        
\begin{tikzpicture}[
level 1/.style={sibling distance=18em},
  level 2/.style={sibling distance=10em},
  level 3/.style={sibling distance=4em},
  level 4/.style={sibling distance=2em},
  level 5/.style={sibling distance=1.2em},
  every node/.style={circle, draw, minimum size=0.5cm, inner sep=0pt},
  path node/.style={fill=white, draw=red, thick, fill=red!20},   
  fixed node/.style={fill=blue!20}, 
  label style/.style={text height=1.5ex, text depth=.25ex},
  level distance=1.2cm
]

\node[path node, label=west:$z_p$] (r) {$r$}
  child {
    node[path node, label=west:$z_3$] { }
    child {
      node[path node, label=west:$z_2$] { }
      child {
        node[path node, label=west:$z_1$] {$1$}
      }
      child {
        node {$2$}
      }
    }
    child {
      node[fixed node]  { }
      child {
        node{$3$}
      }
      child {
        node {$4$}
      }
    }
  }
  child {
    node { }
    child {
      node { }
      child {
        node {$5$}
      }
      child {
        node {$6$}
      }
    }
    child {
      node { }
      child {
        node {$7$}
      }
      child {
        node {$8$}
      }
    }
  };

\end{tikzpicture}
        \caption{Red nodes highlight the path $P$ from the root $r$ to the leaf representing coordinate $1$ in $T$. The blue node is not on $P$ and the input from this node to $z_3$ is the same for both inputs $\vec x$~and~$\vec y$.}
        \label{fig:tree-rounding-path}
    \end{figure}

\begin{lem}\label{lem:tree-rounding}
     \Cref{alg:tree-rounding} is a $2\lceil \log n\rceil$-stretch correlated rounding algorithm for~$\Delta_{n,k}$.
\end{lem}
\begin{proof}
    Let $\vec x, \vec y$ be inputs for \Cref{alg:tree-rounding}, $T$ be the coordinate tree and $\hat{\vec x}, \hat{\vec y}$ be the corresponding outputs of the algorithm. We will show 
    $$\mathbb{E}\left[\ \norm{\hat{\vec x} - \hat{\vec y}}_1\right] \leq 2\lceil \log n\rceil \cdot \norm{\vec x - \vec y}_1.$$

    By \Cref{obs:triangle-ineq}, we can assume input vectors $\vec x, \vec y$ only differ by $\veps$ in the first coordinate. Let $P$ be the path of length $p$ from the leaf representing coordinate $1$ to the root $r$ in $T$, and denote the sequence of nodes in this path $z_1 = 1, \dots, z_p = r$. For any non-leaf node $z_t$ along path $P$, $z_t$ receives coordinates from two children, one on $P$ and one not. The input from the child not on $P$ is always the same for both inputs, see \Cref{fig:tree-rounding-path}. It follows that the contribution to $\mathbb{E}\left[\norm{\hat{\vec x} - \hat{\vec y}}_1\right]$ from coordinates greater than $p$  is zero.

    Let $\vec x^t, \vec y^t \in [0, 1]^p$ be the vectors of marginals of the first $p$ coordinates after the execution of \Cref{alg:node-rounding} for node $z_t$ for inputs $\vec x$ and $\vec y$ respectively. We only consider the first $p$ coordinates since the execution can only differ along $P$. For $p = \lceil \log n \rceil$, the proof follows from the fact that $\norm{\vec{x}^1 - \vec{y}^1}_1 = \veps$ (by \Cref{obs:triangle-ineq}), and by summing over $p$ the inequality proved in \Cref{lem:tree-distance}.
\end{proof}

\begin{lem}\label{lem:tree-distance} 
For any vectors $\vec x, \vec y$ and $\forall t = 2, \dots, p$,
$$\mathbb{E}\left[\norm{\vec{x}^t - \vec{y}^t}_1 - \norm{\vec{x}^{t-1} - \vec{y}^{t-1}}_1\right] \leq 2 \veps.$$
\end{lem}

\begin{proof}
Fix $t \in \{2, \dots, p\}$ and consider node $z_t$ on path $P$ which takes inputs from $z_{t-1}$ which lies on $P$ and $w_{t-1}$ which does not. Let $i_x$ and $i_x$ denote the coordinates that $z_t$ receives from $z_{t-1}$ for $\vec x$ and $\vec y$ respectively, while $j$ is the coordinate $z_t$ receives from $w_{t-1}$ in both cases. 
From the definitions of \Cref{alg:node-rounding}, \Cref{alg:tree-rounding} and \Cref{prop:tree-properties}, the marginal value of $i_y$ that $z_{t-1}$ sends to $z_t$ in the execution of rounding for $\vec y$ is greater by exactly $\veps$ from the marginal value $i_x$ that $z_{t-1}$ sends to $z_t$ in the execution of rounding for $\vec x$. Hence we can assume that node $z_t$ receives coordinate and marginals $(i_x, \alpha)$ and $(j, \beta)$ for input $\vec x$, while it receives $(i_y, \alpha + \veps)$ and $(j, \beta)$ for input $\vec y$. Let $\vec x^{t-1}$ denote the vector $\vec x$ after performing coordinate tree rounding for $t-1$ iterations and $\vec x_{i}^{t-1}$ its $i^{th}$ coordinate. We have $\norm{\vec{x}^1 - \vec{y}^1}_1 = \veps$. Let $\Delta\triangleq \norm{\vec{x}^t - \vec{y}^t}_1 - \norm{\vec{x}^{t-1} - \vec{y}^{t-1}}_1$. 

There are two cases depending on the value of $\alpha + \beta$, namely $0 \leq \alpha + \beta < \alpha + \beta + \veps \leq 1$ and $1 \leq \alpha + \beta + \veps \leq 2$. Since we can take $\veps$ to be arbitrarily small, it is never the case that $0 \leq \alpha + \beta < 1 < \alpha + \beta + \veps \leq 2$. We will prove that
\[\mathbb{E}\left[\norm{\vec{x}^t - \vec{y}^t}_1 - \norm{\vec{x}^{t-1} - \vec{y}^{t-1}}_1\right] = 
		\begin{cases}
  			\veps \frac{2 \beta}{\alpha + \beta + \veps},  & 0 \leq \alpha + \beta < \alpha + \beta + \veps \leq 1, \\
  			\veps \frac{2(1 - \beta)}{2 - \alpha - \beta}, & 1 \leq \alpha + \beta + \veps \leq 2.
		\end{cases}\] 
This concludes the proof, since $\frac{2 \beta}{\alpha + \beta + \veps} \leq 2$ and $\frac{2(1 - \beta)}{2 - \alpha - \beta} \leq 2$ when $0 \leq \alpha, \beta \leq 1$.

\paragraph{Case 1: $\alpha$ and $\beta$ satisfy $0 \leq \alpha + \beta < \alpha + \beta + \veps \leq 1$.} There are three possible outcomes for the values $\vec x^t$ and $\vec y^t$. Case (i) occurs when $\vec x^t, \vec y^t$ are rounded down, case (ii) occurs when $\vec x^t$ is rounded down and $\vec y^t$ is rounded up, and case (iii) occurs when  $\vec x^t, \vec y^t$ are rounded up. The probability of each of these cases is given in \Cref{table:prob-case1}.

\renewcommand{\arraystretch}{1.4}
\begin{table}[h!]
\centering
\begin{tabular}{|c|c|c|c|}
\hline
Case & $\vec{x}^t$ values & $\vec{y}^t$ values & Probability \\
\hline
(i) & $\vec{x}_{i}^t = \alpha + \beta$, $\vec{x}_{j}^t = 0$ & $\vec{y}_{i}^t = \alpha + \beta + \veps$, $\vec{y}_{j}^t = 0$ & $\frac{\alpha}{\alpha + \beta}$ \\
(ii) & $\vec{x}_{i}^t = 0$, $\vec{x}_{j}^t = \alpha + \beta$ & $\vec{y}_{i}^t = \alpha + \beta + \veps$, $\vec{y}_{j}^t = 0$ & $\frac{\beta \veps}{(\alpha + \beta)(\alpha + \beta + \veps)}$ \\
(iii) & $\vec{x}_{i}^t = 0$, $\vec{x}_{j}^t = \alpha + \beta$ & $\vec{y}_{i}^t = 0$, $\vec{y}_{j}^t = \alpha + \beta + \veps$ & $\frac{\beta}{\alpha + \beta + \veps}$ \\
\hline
\end{tabular}
\caption{Probabilities of all possible outcomes for rounding $\vec x^t$ and $\vec y^t$ for Case 1.}
\label{table:prob-case1}
\end{table}

Given $\vec x^t, \vec y^t$,  $\Delta$ is determined by the values of $\vec x^{t-1}_{i_x}$ and $\vec y^{t-1}_{i_y}$. There are two cases, depending on the coordinate $i$.

\paragraph{1.} When $i = i_x = i_y$, $\vec x^{t-1}_{i_x} = \vec x^{t}_{i_x}$ and $\vec y^{t-1}_{i_y} = \vec y^{t}_{i_y}$, so $\lVert \vec{x}^{t-1} - \vec{y}^{t-1} \rVert_1 = \veps$. $\Delta$ takes the following values in cases (i)-(iii). 

\begin{table}[h!]
\centering
\begin{tabular}{|c|c|c|c|}
\hline
 $ $ & Case (i) & Case (ii) & Case (iii) \\
\hline
$\Delta$ & $0$ & $2(\alpha + \beta)$ & $0$ \\
\hline
\end{tabular}
\caption{$\Delta$ values for $i = i_x = i_y$ in Case 1.}
\label{table:delta-case11}
\end{table}

We can directly compute:
\begin{align*}
    \E[\Delta] = \veps \frac{2 \beta}{\alpha + \beta + \veps}.
\end{align*}

\paragraph{2.} When $i_x \neq i_y$, $\Delta$ depends on how the coordinates $i_x$ and $i_y$ were rounded in $\vec x^{t-1}$ and $\vec y^{t-1}$ respectively. There are four subcases (a)-(d). The probabilities for cases (i) - (iii) are the same as in Table 1.

\begin{table}[h!]
\centering
\begin{tabular}{|c|c|c|c|c|}
\hline
$ $ & $\vec x^{t-1}$ and $\vec y^{t-1}$ values & Case (i) & Case (ii) & Case (iii) \\
\hline
(a) & $\vec{x}_{i_y}^{t-1} = 0, \vec{y}_{i_x}^{t-1} = 0$ & $2\beta$ & $2\beta$ & $-2\alpha$ \\
(b) & $\vec{x}_{i_y}^{t-1} = 0, \vec{y}_{i_x}^{t-1} = 1$ & $0$ & $2(\alpha + \beta)$ & $0$ \\
(c) & $\vec{x}_{i_y}^{t-1} = 1, \vec{y}_{i_x}^{t-1} = 0$ & $0$ & $0$ & $2\veps$ \\
(d) & $\vec{x}_{i_y}^{t-1} = 1, \vec{y}_{i_x}^{t-1} = 1$ & $-2\beta$ & $2\alpha$ & $2(\alpha + \veps)$ \\
\hline
\end{tabular}
\caption{$\Delta$ Values for $i_x \neq i_y$ in Case 1.}
\label{table:delta-case12}
\end{table}

We can directly compute:
\begin{align*}
    \mathbb{E}[\Delta] = \frac{\veps (2 \beta)}{\alpha + \beta + \veps}.
\end{align*}

\paragraph{Case 2: $\alpha$ and $\beta$ satisfy $1 \leq \alpha + \beta < \alpha + \beta + \veps \leq 2$.} There are three possible outcomes for the values $\vec x^t$ and $\vec y^t$. Case (i) occurs when $\vec x^t, \vec y^t$ are rounded down, case (ii) occurs when $\vec x^t$ is rounded down and $\vec y^t$ is rounded up, and case (iii) occurs when  $\vec x^t, \vec y^t$ are rounded up. The probability of each of these cases is given in \Cref{table:prob-case2}.

\begin{table}[h!]
\centering
\begin{tabular}{|c|c|c|c|}
\hline
Case & $\vec{x}^t$ values & $\vec{y}^t$ values & Probability \\
\hline
(i) & $\vec{x}_{i_x}^t = 1$, $\vec{x}_{j}^t = \alpha + \beta - 1$ & $\vec{y}_{i_y}^t = 1$, $\vec{y}_{j}^t = \alpha + \beta + \veps - 1$ & $\frac{1 - \beta}{2 - \alpha - \beta}$ \\
(ii) & $\vec{x}_{i_x}^t = \alpha + \beta - 1$, $\vec{x}_{j}^t = 1$ & $\vec{y}_{i_y}^t = 1$, $\vec{y}_{j}^t = \alpha + \beta + \veps - 1$ & $\frac{\veps(1 - \beta)}{(2 - \alpha - \beta)(2 - \alpha - \beta - \veps)}$ \\
(iii) & $\vec{x}_{i_x}^t = \alpha + \beta - 1$, $\vec{x}_{j}^t = 1$ & $\vec{y}_{i_y}^t = \alpha + \beta + \veps - 1$, $\vec{y}_{j}^t = 1$ & $\frac{1 - \alpha - \veps}{2 - \alpha - \beta - \veps}$ \\
\hline
\end{tabular}
\caption{Probabilities of all possible outcomes for rounding $\vec x^t$ and $\vec y^t$ for Case 2.}
\label{table:prob-case2}
\end{table}

Given $\vec x^t, \vec y^t$,  $\Delta$ is determined by the values of $\vec x^{t-1}_{i_x}$ and $\vec y^{t-1}_{i_y}$. There are two cases, depending on the coordinate $i$.

\paragraph{1.} When $i = i_x = i_y$, $\vec x^{t-1}_{i_x} = \vec x^{t}_{i_x}$ and $\vec y^{t-1}_{i_y} = \vec y^{t}_{i_y}$, so $\lVert \vec{x}^{t-1} - \vec{y}^{t-1} \rVert_1 = \veps$. $\Delta$ takes the following values in cases (i)-(iii). 

\begin{table}[h!]
\centering
\begin{tabular}{|c|c|c|c|}
\hline
 $ $ & Case (i) & Case (ii) & Case (iii) \\
\hline
$\Delta$ & $0$ & $2(2 - \alpha - \beta - \veps)$ & $0$ \\
\hline
\end{tabular}
\caption{$\Delta$ values for $i = i_x = i_y$ in Case 2.}
\label{table:delta-case21}
\end{table}

We can directly compute:
\begin{align*}
    \mathbb{E}[\Delta] = \veps \frac{2(1 - \beta)}{2 - \alpha - \beta}.
\end{align*}

\paragraph{2.} When $i_x \neq i_y$, $\Delta$ depends on how the coordinates $i_x$ and $i_y$ were rounded in $\vec x^{t-1}$ and $\vec y^{t-1}$ respectively. The probabilities for cases (i) - (iii) are the same as in Table 4. 

\begin{table}[h!]
\centering
\begin{tabular}{|c|c|c|c|c|}
\hline
Case & $\vec x^{t-1}$ and $\vec y^{t-1}$ values & Case (i) & Case (ii) & Case (iii) \\
\hline
(a) & $\vec{x}_{i_y}^{t-1} = 0, \vec{y}_{i_x}^{t-1} = 0$ & $2(1 - \alpha)$ & $2(1 - \alpha - \veps)$ & $-2(1 - \beta)$ \\
(b) & $\vec{x}_{i_y}^{t-1} = 0, \vec{y}_{i_x}^{t-1} = 1$ & $0$ & $2(2 - \alpha - \beta - \veps)$ & $0$ \\
(c) & $\vec{x}_{i_y}^{t-1} = 1, \vec{y}_{i_x}^{t-1} = 0$ & $2\veps$ & $0$ & $0$ \\
(d) & $\vec{x}_{i_y}^{t-1} = 1, \vec{y}_{i_x}^{t-1} = 1$ & $-2(1 - \alpha - \veps)$ & $2(1 - \beta)$ & $2(1 - \beta)$ \\
\hline
\end{tabular}
\caption{$\Delta$ Values when $i_x \neq i_y$ in Case 2.}
\label{table:delta-case22}
\end{table}

We can directly compute:
\begin{align*}
    \mathbb{E}[\Delta] = \veps \frac{2(1 - \beta)}{2 - \alpha - \beta}.
\end{align*}

This completes the proof.
\end{proof}

\begin{obs}\label{obs:fractional-bound-tree-rounding}
    \Cref{alg:tree-rounding} applied to $\vec{x},\vec{y}$ has stretch at most $2(frac(\vec{x})+frac(\vec{y}))$, for $frac(\vec{z})=|\{i\mid z_i\not\in\{0,1\}\}|$ the number of fractional coordinates of $z$. 
\end{obs}

\subsubsection{Input-Sparsity Time and Other Computational Considerations}

In this section we consider the time needed to implement \Cref{alg:tree-rounding} in input-sparsity time, and logarithmic parallel depth and update time. Since the latter two are quite straightforward, we defer them to the end of this section, and focus now on input-sparsity time implementation.

Naively, the algorithm of \cite{chen2017correlated} requires $O(n)$ time --- proportional to the number of nodes in the balanced tree used.
However, much of the computation and random bits used in this tree are unnecessary if the input vector $\vec{x}$ is sparse.
Intuitively, one can skip work in internal nodes of \Cref{alg:node-rounding} for nodes that do not have a leaf with non-zero initial value in both their subtrees. 
That is, we would ideally like to ``shortcut'' this step in internal nodes which are not the lowest-common ancestor (LCA) of two leaves with initially non-zero value.
Computing such LCAs will therefore prove useful shortly. As we now show, this task is simple to do in $O(1)$ time using simple bit tricks.\footnote{The following lemma is undoubtedly known, but we could not find a reference for it. We do not claim any novelty for it, and prove it only for completeness.}

\begin{figure}
        \centering
        \begin{tikzpicture}
         [level distance=1cm, 
    level 1/.style={sibling distance=10em},
      level 2/.style={sibling distance=8em},
      level 3/.style={sibling distance=4em}, level distance=1.2cm, every node/.style={circle, draw, minimum size=0.5cm, inner sep=1pt},]
        \node[label=west:$001$] {$1$}
            child { node[label=west:$010$] {$2$}
                child { node[label=west:$100$] {$4$} }
                child { node[label=west:$101$] {$5$} }
            }
            child { node[label=east:$011$] {$3$}
                child { node[label=east:$110$] {$6$} }
                child { node[label=east:$111$] {$7$} }
            };
    \end{tikzpicture}
        \caption{The standard binary encoding of nodes where the root is labeled with $1$. The left child of node $i$ is labeled $2i$ and the right child is labeled $2i + 1$.}
        \label{fig:standard-encoding}
    \end{figure}
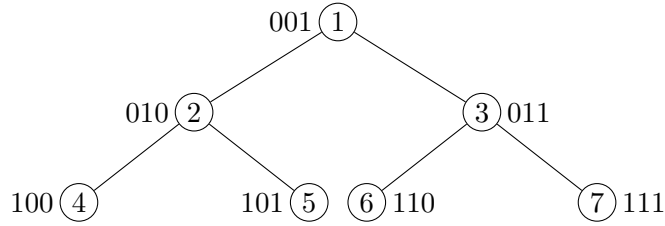

    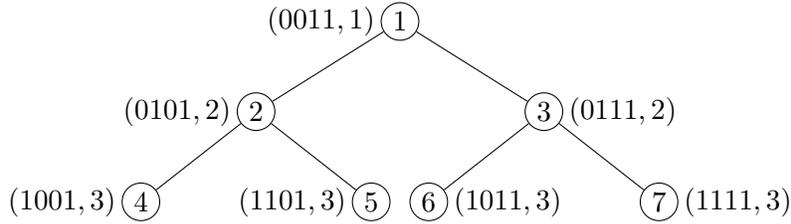
\begin{figure}
    \centering
    \begin{tikzpicture}
    [level distance=1cm, 
    level 1/.style={sibling distance=10em},
      level 2/.style={sibling distance=8em},
      level 3/.style={sibling distance=4em}, level distance=1.2cm, every node/.style={circle, draw, minimum size=0.5cm, inner sep=1pt},]
        \node[label=west:${(0011, 1)}$] {$1$}
            child { node[label=west:${(0101, 2)}$] {$2$}
                child { node[label=west:${(1001, 3)}$] {$4$} }
                child { node[label=west:${(1101, 3)}$] {$5$} }
            }
            child { node[label=east:${(0111, 2)}$] {$3$}
                child { node[label=east:${(1011, 3)}$] {$6$} }
                child { node[label=east:${(1111, 3)}$] {$7$} }
            };
    \end{tikzpicture}
    \caption{Binary tree where node $i$ is associated with tuple $Encode(i)$.}
    \label{fig:binary-tree-labels}
    \end{figure}

\begin{lemma}
    There exists an encoding of the nodes of a complete binary tree with $n$ leaves that can be stored in $O(n)$ space such that the parent and children of any node can be computed in $O(1)$ time and the LCA of two nodes can be computed in $O(1)$, time using standard RAM operations. 
\end{lemma}
\begin{proof}
    The standard encoding of complete binary trees in  array-based representations of binary heaps \cite{cormen2022introduction} is as follows: Encode the root as $1$ and  all children of node encoded with $i$ is represented using $2i$ and $2i+1$ for left and right child, respectively. This is shown in \Cref{fig:standard-encoding}. 

    The LCA of two nodes in the standard binary encoding is given by the longest common subsequence of the most significant bits. In order to use constant-time operations, we will define a new binary encoding that will allow us to find the LCA in terms of least significant bits. To represent the complete binary tree, associate every node $i$ with $Encode(i) = (2^{L} + \text{reverse}(\text{binary encoding}(i)), L)$. $L$ is the level of $i$, given by $\lfloor \log_2(i) \rfloor + 1$. See \Cref{fig:binary-tree-labels} for an illustration of the encoding scheme. 
    
    Every subsequence of most significant bits of $i$ in the standard encoding is equivalent to the subsequence of the same length of the least significant bits in $Encode(i)$. The encoding scheme allows us to represent a complete binary tree with $n$ leaves in $O(n)$ space. Given a node $i$, the parent of $i$ and the left and right children of $i$ can be computed in $O(1)$ time using bit-wise operations, as addition, subtraction and exponentiation can be done with bit-wise operations right shift, $AND, OR, NOT$. 
    \begin{algorithm}[H]
	\caption{Parent($i, L$)}
	\label{alg:parent}
 \begin{algorithmic}[1]
    \State \textbf{Output} $(i - 2^{L+1} + (i$  $AND$  $2^L == 0) \cdot 2^L, L-1)$
    \end{algorithmic}	
\end{algorithm}	

    \begin{algorithm}[H]
    	\caption{Children($i, L$)}
    	\label{alg:child}
     \begin{algorithmic}[1]
        \State $Right \gets i + 2^{L+1}$.
        \State $Left \gets i - 2^L + 2^{L+1}$. 
        \State \textbf{Output} $(Right, L+1), (Left, L+1)$.
        \end{algorithmic}	
    \end{algorithm}	

    The LCA of two nodes can be found in $O(1)$ time using the following algorithm that finds the longest common subsequence of least significant bits. See \Cref{fig:lca-computation} for an illustration.

    \begin{algorithm}[H]
	\caption{LCA($(x, L_x),(y, L_y)$)}
	\label{alg:lca}
     \begin{algorithmic}[1]
        \State $z\gets XOR(x,y)$.
        \State $c\gets AND(z,NOT(z)+1)$. 
        \State \textbf{Output} $(AND(x,c-1) + c, \log_2(c))$.
        \end{algorithmic}	
    \end{algorithm}	

    This algorithm clearly runs in $O(1)$ time (given that all bit-wise operations it uses do). The first line computes a value $z$ whose least significant non-zero bit is the least significant bit for which $x$ and $y$ differ. The next line is a classic exercise to compute $c=2^i$, where $i$ is the first non-zero bit of $z$. It uses that $NOT(z) + 1$ has all bits below $i$ set to zero, the $i^{th}$ bit set to one, and all higher bits being the bit-wise negation of their counterparts in $z$. Finally, $c-1=2^i-1$ is then a mask for the least significant bits of $x$ and $y$ that are equal. The addition of $c$ appends the leading $1$ that gives the desired LCA of $x$ and $y$, represented in the form of the $Encode$ function. 
\end{proof}

\begin{figure}
    \centering
    \begin{tikzpicture}
  \def\cellsize{0.8}
  \def\rowsep{1.2}

  \def\binaries{{1, 1, 1, 0, 1, 0, 1}, {1, 0, 1, 1, 1, 0, 1},  {0, 1, 0, 1, 0, 0, 0}, {1, 0, 1, 1, 0, 0, 0}, {0, 0, 0, 1, 0, 0, 0}, {0, 0, 0, 0, 1, 1, 1}, {0 ,0, 0, 1, 1, 0, 1}}

  \def\labels{{"x", "y", "z", "NOT(z) + 1", "c", "c-1", "AND(x, c-1) + c"}}

  \def\arraysPerCol{8}

  \foreach \bin [count=\idx from 0] in \binaries {
    \pgfmathtruncatemacro{\dec}{\idx + 1}
    \pgfmathtruncatemacro{\col}{int(\idx / \arraysPerCol)}
    \pgfmathtruncatemacro{\row}{int(mod(\idx,\arraysPerCol))}

    \pgfmathsetmacro{\xstart}{\col * 5 * \cellsize} 
    \pgfmathsetmacro{\ystart}{-\row * \rowsep}

    \pgfmathparse{\labels[\idx]}
    \let\labeltext\pgfmathresult
    \node[right] at ({\xstart - 2.9}, {\ystart - \cellsize/2}) {\small \labeltext};

    \foreach \bit [count=\bitidx from 0] in \bin {
      \pgfmathsetmacro{\x}{\xstart + \bitidx * \cellsize}
      \node[draw, minimum width=\cellsize cm, minimum height=\cellsize cm, anchor=north west]
        at (\x, \ystart) {\bit};
    }
  }
\end{tikzpicture}
    \caption{\Cref{alg:lca} on binary strings $x = 1110101$ and $y = 1011101$ to find their LCA, $1101$.}
    \label{fig:lca-computation}
\end{figure}

\begin{rem}
    Note that we assume that $n$  fits in a single machine word (or $O(1)$ such), otherwise computing the LCA can be done in $O(\log n)$ time using binary search. 
    However, similar slow downs then occur due to the time to read pointers for example. 
    We ignore this point and stick to the word RAM model of computation with word size at least $\log n$ bits (necessary for pointers, etc).
\end{rem}

Returning to the implementation of pivotal sampling (\Cref{alg:tree-rounding}) in input-sparsity time, we only need to perform (non-trivial) versions of \Cref{alg:node-rounding} in internal nodes $v$ with non-zero coordinates represented by leaves in both left and right subtrees. 
That is, only nodes that are LCAs of two consecutive such leaves, whose indices we denote by $\ell_1<\ell_2<\dots<\ell_{nnz}$.
The following algorithm performs the required postorder traversal of the compressed graph containing only these internal nodes and leaves (after compressing nodes with no such non-zero leaf in one of their subtrees). See \Cref{fig:sparse-algorithm} for an illustration of this algorithm. 

\begin{algorithm}[H]
	\caption{Input-sparsity time pivotal sampling}
	\label{alg:nnz-pivotal-sampling}
 \begin{algorithmic}[1]
    \State $S\gets \emptyset$. \Comment{$S$ is a stack}
    \For{$i=1,\dots,nnz-1$}
    \State $(w, L_w) \gets LCA((\ell_i, L_i), (\ell_{i+1}, L_{i+1}))$. 
    \While{$S\neq \emptyset$ and $LCA(S.top, w)\neq S.top$}
    \State $(w', \ell_{j}, \ell_{k}) \gets S.pop$.
    \State $\alpha_j \gets $ Non-fixed marginal at $\ell_{j}$.
    \State $\alpha_{k} \gets $ Non-fixed marginal at $\ell_{k}$.
    \State $(\ell_j, \alpha_j', \ell_{k}, \alpha_{k}', s) \gets $ Resolve$(\ell_j, \alpha_j, \ell_{k}, \alpha_k)$.
    \State Update the marginals at $w'$ with $\alpha_j', \alpha_k'$.
    \EndWhile 
    \If{$S=\emptyset$ \textbf{or} $S.top\neq w$}
    \State $S.push(w, \ell_i, \ell_{i+1})$. 
    \EndIf
    \EndFor
    \end{algorithmic}	
\end{algorithm}

\begin{lem}
    \Cref{alg:nnz-pivotal-sampling} is an $2\lceil \log n\rceil = O(\log n)$-stretch $O(nnz)$ sample time correlated sampling algorithm for $\Delta_{n,k}$.
\end{lem}
\begin{proof}
    Each internal node $v$ that is an LCA of two consecutive non-zero leaves is clearly pushed to the stack at some point (during iteration $i$ for $\ell_i$ the right-most non-zero leaf in $v$'s left subtree and $\ell_{i+1}$ the left-most non-zero leaf in $v$'s right subtree). Moreover, as we prove below, each such internal node is popped after all nodes in its subtree have been processed. The stretch then follows from \Cref{lem:tree-rounding}, together with all nodes for which we implicitly perform \Cref{alg:node-rounding} leaves the zero and non-zero values unchanged.
    The running time follows from LCA and Resolve both taking $O(1)$ time, and the number of nodes processed being precisely $2nnz-1=O(nnz)$, by the above.
    
    It remains to prove that each internal node is processed after all nodes in its subtree. To prove this, we prove by induction that the stack maintains the ancestry invariant, that is the nodes in the stack appear in ancestor-to-descendant order from bottom to top, and no node is popped before its descendants. This invariant holds vacuously before any iterations occurred, since the stack is empty.
    For the inductive step, assume that after $i - 1$ iterations, the ancestry variant is maintained and consider iteration $i$, where we process pair $(\ell_i, \ell_{i+1})$. Let $w = \mathrm{LCA}(\ell_i, \ell_{i+1})$. If the stack is empty we push $w$, and the invariant holds. 
    Otherwise, let $S.\mathrm{top}$ be the node at the top of the stack. By definition, $LCA(S.top, w) = S.top$ if and only if $w$ is a descendant of $S.top$. Hence, if $w$ is not a descendant of $S.\mathrm{top}$, then $\mathrm{LCA}(S.top, w) \neq w$ and the condition for the while loop is satisfied. We pop nodes from the stack until $S.top$ is an ancestor of $w$ (or the stack becomes empty). When $w$ is pushed onto the stack, this maintains the ancestry invariant. 
\end{proof}

\paragraph{Parallel Implementation.} Using a (near-)complete binary tree on $n$ leaves, it is easy to implement \Cref{alg:tree-rounding} using $n$ processors and depth $\log n + 1$, by iteratively performing \Cref{alg:node-rounding} for all nodes of depth $\log n + 1, \log n, \log n - 1,\dots, 1$. We thus obtain the following from \Cref{lem:tree-rounding}.
\begin{fact}\label{fact:parallel-ICALP17}
    There exits an $O(\log n)$-stretch correlated sampling algorithm for  $\Delta_{n,k}$ with parallel depth $O(\log n)$ and $O(n)$ work.
\end{fact}

\paragraph{Dynamic Implementation.}
Similar to the above, for each change in $\vec{x}$ in coordinate $i$, one can simply perform the changes along the leaf-to-root path from the leaf corresponding to coordinate $i$. This yields the same output as recomputation from scratch, while using $O(\log n)$ \emph{update time}, or time per change to the input vector. This update time is precisely the number of levels time the time needed for \Cref{alg:node-rounding}.

\begin{fact}\label{fact:dynamic-ICALP17}
    There exits an $O(\log n)$-stretch correlated sampling algorithm for $\Delta_{n,k}$ with $O(\log n)$ update time.
\end{fact}

\begin{figure} 
\centering
\begin{tikzpicture}  
[level 1/.style={sibling distance=18em},
  level 2/.style={sibling distance=10em},
  level 3/.style={sibling distance=4em},
  level 4/.style={sibling distance=2em},
  level 5/.style={sibling distance=1.2em},
  every node/.style={circle, draw, minimum size=0.5cm, inner sep=0pt, font=\bfseries},
  empty node/.style={fill=white},   
  full node/.style={fill=red}, 
  leaf node/.style={fill=blue},  
  label style/.style={text height=1.5ex, text depth=.25ex},  
  level distance=1.2cm 
  ]

\node[full node, label=west:$5$] (n1) { }
    child {node[full node, label=west:$2$] (n2) {}
      child {node[empty node] (3) { }
        child {node[full node, label=west:$1$] (n4) {}
            child {node[leaf node, label=below:$\ell_1$] (n5) { }}
            child {node[leaf node, label=below:$\ell_2$] (n6) { }}}
        child {node[empty node] (n7) { }
            child {node[empty node] (n8) { }}
            child {node[empty node] (n9) { }}}
      }
      child {node[full node, label=west:$3$] (n10) { }
        child {node[empty node] (n11) { }
            child {node[empty node] (n12) { }}
            child {node[leaf node, label=below:$\ell_3$] (n13) { }}}
        child {node[full node, label=west:$4$] (n14) { }
            child {node[leaf node, label=below:$\ell_4$] (n15) { }}
            child {node[leaf node, label=below:$\ell_5$] (n16) { }}}
      }
    }
    child {node[empty node] (n17) { }
      child {node[full node, label=west:$7$] (n18) { }
        child {node[full node, label=west:$6$] (n19) { }
            child {node[leaf node, label=below:$\ell_6$] (n20) { }}
            child {node[leaf node, label=below:$\ell_7$] (n21) { }}}
        child {node[empty node] (n22) { }
            child {node[leaf node, label=below:$\ell_8$] (n23) { }}
            child {node[empty node] (n24) { }}}
      }
      child {node[empty node] (n25) { }
        child {node[empty node] (n26) { } 
            child {node[empty node] (n27) { }}
            child {node[empty node] (n28) { }}}
        child {node[empty node] (n29) { }
            child {node[empty node] (n30) { }}
            child {node[empty node] (n31) { }}}
      }
    };
\end{tikzpicture}
\caption{\Cref{alg:node-rounding} is called at internal nodes (colored red) that have non-zero valued leaves (colored blue) in both subtrees when these internal nodes are popped. The order in which the internal nodes are popped in the above example is $1, 4, 3, 2, 6, 7, 5$.}
\label{fig:sparse-algorithm}
\end{figure}
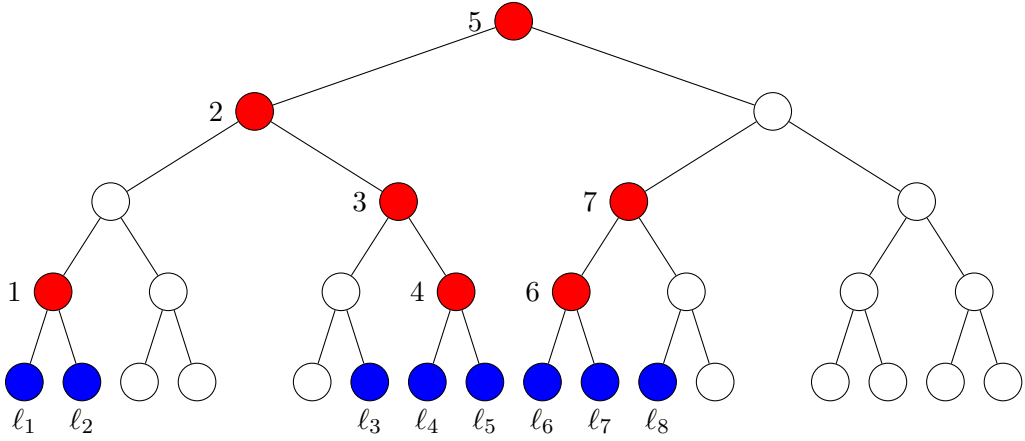

\section{Extension to Partition Matroids}\label{sec:partition}

We note that any rounding scheme for cardinality constraints (a.k.a.~uniform matroids) extends trivially to partition matroids. 
Given a rounding algorithm $\mathcal{A}$ which maintains cardinality constraints and marginals, one can naturally apply this rounding algorithm independently to each part $E_i$, $i=1,...,p$, with the respective cardinality constraint $k_i$. This trivially preserves the partition matroid's (fractional) constraint, and the marginals, allowing us to preserve linear objectives. We show that if $\mathcal{A}$ preserves submodular objectives with respect to the multilinear extensions $F:[0,1]^{|E|}\to\mathbb{R}$ of any submodular function $f:2^{E}\to \mathbb{R}$, then so does the derived algorithm $\mathcal{A}'$.
\begin{lem}\label{lem:uniform-to-partition}
    If $\mathcal{A}$ satisfies submodular dominance, then so does $\mathcal{A}'$.
\end{lem}
\begin{proof}
    We prove by induction on $p\geq 1$ that the above holds for any ground-set $E$, submodular function $f:2^{E}\to\mathbb{R}$ and partition of the ground-set into $p$ parts that $\E[f(\calA(\vec{x}))]\geq F(\vec{x})$. The base case $p=1$ is trivial, since $\mathcal{A}'=\mathcal{A}$ in this case.
    
For the inductive step, for $\vec{x}\in [0,1]^E$, denote by $\vec{x}_{-p}\in [0,1]^{E\setminus E_p}$ the restriction of $\vec{x}$ to $E\setminus E_p$.
For fixed $S_p\subseteq E_p$, consider the submodular function $g_{S_p}:[0,1]^{E\setminus E_p}\to \mathbb{R}$ with $g_{S_p}:S \mapsto f(S\cup S_p)$, and denote its multilinear extension by $G_{S_p}:[0,1]^{E\setminus E_p}\to \mathbb{R}$.
By the inductive hypothesis, since $g_{S_p}$ is defined over a ground-set with $p-1$ parts, we have
$$\E[f(\mathcal{A}'(\vec{x})) \mid \mathcal{A}'(\vec{x})\cap E_p = S_p] = \E[g_{S_p}(\mathcal{A}'(\vec{x}_{-p}))] \geq G_{S_p}(\vec{x}_{-p}).$$
Next, denote by $\vec{x}_{\mid E_p}$ the restriction of $\vec{x}$ to $E_p$, and 
define the submodular function $h:2^{E_p}\to \mathbb{R}$ with $h:S_p \mapsto \E[f(\mathcal{A}'(\vec{x})) \mid \mathcal{A}'(\vec{x})\cap E_p = S_p]$.
Then, since $\calA'(\vec{x})\setminus E_p$ and $\calA'(\vec{x})\cap E_p$ are independent, the function $h$ is submodular.\todo{don't feel like unpacking this} Therefore, since $E_p$ has a single part, 
\begin{align*}
    \E[f(\mathcal{A}'(\vec{x}))] = \E[h(\mathcal{A}(\vec{x}_{\mid E_p}))] \geq H(\vec{x}_{\mid E_p}).
\end{align*}
Combining the definitions of the multilinear extension, $f$ and $g_{S_p}$ with the preceding inequalities, we obtain the desired inequality:
\begin{align*}
    \E[f(\mathcal{A}'(\vec{x}))] & \geq H(\vec{x}_{\mid E_p}) \\
    & =  \sum_{S_p\subseteq E_p} h(S_p) \prod_{e\in S_p}x_e \prod_{e\in E_p\setminus S_p} (1-x_e) \\
    & =  \sum_{S_p\subseteq E_p} \E[f(\mathcal{A}'(\vec{x}) \mid \mathcal{A}'(\vec{x})\cap E_p = S_p]  \prod_{e\in S_p}x_e \prod_{e\in E_p\setminus S_p} (1-x_e) \\
    & \geq  \sum_{S_p\subseteq E_p} G_{S_p}(\vec{x}_{-p}) \prod_{e\in S_p}x_e \prod_{e\in E_p\setminus S_p} (1-x_e) \\
    & =  \sum_{S_p\subseteq E_p} \left(\sum_{s\subset E\setminus E_p} f(S\cup S_p) \prod_{e\in S}x_e \prod_{e\in E\setminus (S\cup E_p)} (1-x_e)\right) \prod_{e\in S_p}x_e \prod_{e\in E_p\setminus S_p} (1-x_e) \\
    & =  \sum_{T\subset E} f(T) \prod_{e\in T}x_e \prod_{e\in E\setminus T}(1-x_e) \\
    & =  F(\vec{x}). \qedhere 
\end{align*}
\end{proof}
\section{Dynamic SWF: Solving the LP}\label{app:bnw}

Here we provide ideas needed to solve the LP in \Cref{sec:swf} subject to high submodular values for each scenario. This appendix relies heavily on ideas of \cite{buchbinder2025chasing}, and we claim no novelty for this section.

For our discussion we need another extension of submodular functions to real-valued vectors, namely the \emph{coverage extension} \cite{levin2022submodular}, first studied by Wolsey in the context of submodular set coverage \cite{wolsey1982analysis}.  
\begin{Def}\label{def:wolsey}
    The \emph{coverage extension} $f^*:[0,1]^E\to \mathbb{R}$ of a set function $f:2^E\to \mathbb{R}$ is given by $$f^*(\vec{x})=\min_{S\subseteq E} \bigg(f(S) + \sum_{i\in E} f(i\mid S)\cdot x_i\bigg).$$
\end{Def}
Our interest in this extension is its tight connection to the multilinear extension for monotone submodular functions. 
Vondr\'ak \cite{vondrak2007submodularity} proved that the multilinear extension of any monotone submodular function $(1-1/e)$-approximates the coverage extension.

\begin{prop}\label{lem:F>=(1-1/e)f*}
    Let $f:2^E\to \mathbb{R}_+$ be a monotone subdmodular function and $\vec{x}\in [0,1]^E$ be a vector. Then,
    $F(\vec{x}) \geq F(\vec{1}-\exp(-\vec{x}))\geq (1-1/e) \cdot f^*(\vec{x})$.
\end{prop}
By unboxing the above proof, \cite{buchbinder2025chasing} show that this connection is in fact tight, and algorithmic: if $F(\vec{x})$ is small, then one can find a set witnessing that $f^*(\vec{x})$ is small. Their key technical theorem here is the following.

\begin{restatable}{prop}{separateorapproximate}\label{prop-mainsep}
    Let $f\colon 2^{E} \rightarrow \R_+$ be a normalized monotone submodular function and $\vec{x}\in [0,1]^n$. Let $\eps\in (0,1)$ and $V\geq 0$. If $F(\vec{1}-\exp(-\vec{x}))<(1-1/e-\eps) \cdot V$, then with probability $\Omega(\eps^2)$, 
    the random set $S \triangleq \{i\in E \mid  Y_i \leq t\}$, for independent $t\sim \Uni[0,1]$ and $Y_i\sim \Exp(x_i)$ $\forall i\in E$, satisfies     \begin{equation*}
        f(S)+ \sum_{i\in E}f(i\mid S) \cdot x_i \leq \left(1-\frac{\eps}{4e}\right) \cdot V = (1-O(\eps)) \cdot V.
    \end{equation*}
\end{restatable}

The proof effectively reverse engineers the proof of the preceding proposition by \cite{vondrak2007submodularity}. 
\cite{buchbinder2025chasing} use the preceding propositions to show that $(1-1/e-\eps)$-approximating a single monotone submodular function subject to a separable polytope can then be reduced to the use of cutting plane methods, such as the ellipsoid method: 
At each point in time check for violating constraints of the constraint polytope (and continue the run of the cutting plane method appropriately if one is found); if none is found, then either we have a feasible point with multilinear value that $(1-1/e-\eps)$-approximates the optimal, or we find a separating hyperplane for the constraint $f^*(\vec{x})\geq f(OPT)$. We can therefore continue the cutting plane method until we terminate with a point satisfying $f^*(\vec{x})\geq f(OPT)$, and hence by \Cref{lem:F>=(1-1/e)f*}, $F(\vec{x})\geq (1-1/e)\cdot f(OPT)$, or alternatively we terminate early with a point satisfying the above.

Extending the above for our setting of \Cref{sec:swf}, where we wish many of the sub-vectors to have high multilinear objective value, immediately gives the following.

\begin{thm}
    There exists a polynomial-time algorithm computing a fractional point $\vec{x}$ optimizing \eqref{MMD-LP} subject to the additional constraints $F(\vec{x}_{\mid i})\geq (1-1/e-\eps)\cdot SWF(i)$ for all $i$. The solution's value has value $R$ at most equal to the minimum number of changes between allocations of optimal allocations for each scenario $i$.
\end{thm}
\begin{proof}[Proof (Sketch)]
    We compute a $(1-1/e)$-approximation of the optimal allocation for each configuration. Denote this by $V_i$. We now add to \eqref{MMD-LP} the constraint $f_i(\vec{x}_{\mid i})\geq V_i$.
    Note that the optimal (integral) allocations for configurations satisfy the above LP, and so the objective value is at most $R$.
    Now, using \Cref{prop-mainsep}, using a cutting planes method, in polynomially many steps we either find a violating constraint, or we output a solution $\vec{x}$ such that $F_i(\vec{x}_\mid i)\geq (1-1/e-\eps)\cdot V_i\geq (1-1/e-\eps)^2\cdot SWF(i)$.
\end{proof}
\begin{rem}
    By replacing $V_i$ with $V_i/\beta$ for $\beta>1$, we can trade off approximation for number of changes, decreasing the approximation quality by a factor of $\beta$ while decreasing the number of changes (provided the latter is smaller for near-optimal solutions).
\end{rem}

\bibliographystyle{alpha}
\bibliography{abb,bib,ultimate}

\end{document}